\documentclass[letter,12pt,onecolumn,draftclsnofoot]{IEEEtran}
\renewcommand{\baselinestretch}{1.5}
\IEEEoverridecommandlockouts
\usepackage[pdftex]{graphicx}
 \usepackage{lipsum,graphicx,subcaption}
\usepackage{float}
\usepackage{amsmath,amssymb,amsfonts}
\usepackage{textcomp}
\usepackage{url}

\usepackage[T1]{fontenc}
\usepackage{mathtools, cuted}
\usepackage{lipsum, color}
\usepackage{amsthm}
\newtheorem{theorem}{Theorem}

\newtheorem{lemma}{Lemma}

\newtheorem{remark}{Remark}
\usepackage{mathtools}
\usepackage{graphics}
\DeclarePairedDelimiter\norm{\lVert}{\rVert}
\usepackage{algpseudocode}
\usepackage{float}
\usepackage{verbatim}
\usepackage{array}
\usepackage{cite}
\usepackage{lipsum}
\usepackage{mathtools}
\usepackage{icomma}
\usepackage{optidef}
\usepackage{subcaption}
\usepackage[font=small]{caption}
\usepackage{mwe}
\usepackage{textgreek}
\usepackage{tabularx}
\usepackage{multirow}
\usepackage{stackengine}
\newcommand{\overbar}[1]{\mkern 1.5mu\overline{\mkern-1.5mu#1\mkern-1.5mu}\mkern 1.5mu}
\usepackage{subcaption}
\newcommand{\widesim}[2][1.5]{
  \mathrel{\overset{#2}{\scalebox{#1}[1]{$\sim$}}}
}

\ifCLASSINFOpdf
\else
\fi

\begin{document}

\title{\textcolor{black}{Impact of Beam Misalignment on Hybrid Beamforming NOMA for mmWave Communications}}

\author{
    \IEEEauthorblockN{Mojtaba~Ahmadi~Almasi,   Mojtaba~Vaezi,   Hani~Mehrpouyan}\\


\thanks{This project was supported in part by the NSF ERAS under Grant  1642865. This work was presented in part at VTC Fall 2018~\cite{almasi2018non}.

M. A. Almasi and H. Mehrpouyan are with the Department
of Electrical and Computer Engineering,   Boise State University,   Boise,   ID 83725,   USA (e-mail: mojtabaahmadialm@u.boisestate.edu,  hanimehrpouyan@boisestate.edu). 

M. Vaezi is with the Department of Electrical and Computer Engineering,   Villanova University,   Villanova,   PA 19085,   USA. He is also a Visiting Research Collaborator
 at Princeton University (e-mail: mvaezi@villanova.edu).
}
}

\maketitle


\begin{abstract}
This paper analyzes the effect of beam misalignment on rate performance in downlink of hybrid beamforming-based non-orthogonal multiple access (HB-NOMA) systems. First an HB-NOMA framework is designed in multiuser  millimeter wave (mmWave) communications. A sum-rate maximization problem is formulated for HB-NOMA,  and an algorithm is introduced to design digital and analog precoders and efficient power allocation. Then, regarding perfectly aligned line-of-sight (LoS) channels, a lower bound for the achievable rate is derived. Next, when the users experience misaligned LoS or  non-LoS (NLoS) channels, the impact of beam misalignment is evaluated. To this end, a misalignment factor is modeled and each misaligned effective channel is described in terms of the perfectly aligned effective channel parameters and the misalignment factor. Further, a lower bound for the achievable rate is extracted. We then derive an upper bound for the rate gap expression between the aligned and misaligned HB-NOMA systems. The analyses reveal that a large misalignment can remarkably degrade the rate. Extensive numerical simulations are conducted to verify the findings. 
\end{abstract}

\newpage
\section{Introduction}
\textit{Millimeter wave} (mmWave) communications has emerged as one of the key solutions for the fifth-generation (5G) wireless networks. The existence of  large unused spectrum at mmWave band (30-300 GHz) offers the potential for
significant throughput gains. Shorter wavelengths of the mmWave band, on the other hand, allow for the deployment of  large numbers of antenna elements at both the  base station (BS) and  mobile users, which, in turn, enables
mmWave systems to  support higher degrees
of  \textit{multiplexing} gain in the multiple-input multiple-output (MIMO)
and \textit{multiuser MIMO} systems~\cite{r1,  rappaport2014millimeter,  osseiran2014scenarios, 7593259}.
To this end,   the BS needs to  apply some form of \textit{beamforming}.
This beamforming can be done in the baseband,   radio frequency (RF),   or a combination of the two.
While  \textit{baseband beamforming} (fully-digital) offers  a better
control over the entries of the precoding matrix, 
it is unlikely with current semiconductor technologies due to high hardware cost and power consumption. \textit{Analog beamforming} is an alternative to the baseband beamforming which controls the phase of the signal transmitted at each antenna
using analog phase-shifters  implemented
in the RF domain. Fully-analog beamforming which uses one \textit{RF chain},   see,   e.g.~\cite{r5},   can,   however,   support only one data stream. 

In order to transmit multiple streams and keep the hardware complexity and energy consumption low,   by exploiting several RF chains,    \textit{hybrid analog/digital beamforming} mmWave systems are designed~\cite{el2012low,  el2014spatially}. In~\cite{r9} and~\cite{sayeed2011continuous},   the concept of beamspace MIMO is introduced where several RF chains are connected to a lens antenna array via switches. Recently,   multi-beam lens-based reconfigurable antenna MIMO systems have  been proposed to overcome severe path loss and shadowing in mmWave frequencies~\cite{r19,  almasi2018new}. In the aforementioned systems,   each beam is considered to serve only one user. The works  in~\cite{alkhateeb2015achievable} and~\cite{alkhateeb2015limited} show that exploiting hybrid beamforming in multiuser systems achieves a higher spectral efficiency. Also,  ~\cite{almasi2018reconfigurable} enhances the spectral efficiency by supporting several users through multi-beam reconfigurable antenna. Nevertheless,   the number of served users are far less than the number of users envisioned for 5G networks.

\textit{Non-orthogonal multiple access} (NOMA) is another enabling technique for 5G networks that augments the number of users and spectral efficiency in multiuser scenarios~\cite{saito2013system,  saito2013non,  ding2014performance,  higuchi2015non,  dai2015non,  ding2016impact,  shin2017non,  shin2017coordinated}.
Unlike   orthogonal multiple access (OMA) techniques,   such as
time division multiple access (TDMA),   frequency division multiple access (FDMA),   and code division multiple access (CDMA) which can support only one user per time,   frequency,   or code,   respectively,   NOMA can support multiple users in the same time/frequency/code/beam. NOMA can be realized in the code,   power,   or other domains~\cite{MojtabaBook}. In the power domain,   NOMA employs \textit{superposition coding} at the transmitter. This technique exploits the channel gain difference between users to multiplex their signal. Subsequently,   \textit{successive interference cancellation} (SIC) is applied at the receiver such that the user with better channel first decodes the signal of the user with worse channel and  then subtracts it from the
received signal to decode its own signal~\cite{tse2005fundamentals,  saito2013system,  saito2013non,  ding2014performance,  higuchi2015non,  dai2015non,  ding2016impact,  shin2017non,  shin2017coordinated,  MojtabaBook}.  Beside superiority of NOMA over OMA techniques in terms of number of supported users and spectral efficiency, OMA techniques may not be a practical option for mmWave communications~\cite{wei2018multi}. As an example, TDMA, which serves users through orthogonal time slots but the same spectrum, requires precise and fast timing synchronization. This is because symbol rate in 5G network is far higher than the current networks. Therefore, employing TDMA in mmWave 5G network might be challenging. Exploiting FDMA in mmWave 5G network also can bring about implementation issues. In FDMA, the existing large frequency band is divided into several orthogonal frequency bands. It is expected that FDMA to serve all users via the orthogonal bands at the same time slot. However, due to highly directional beams, the current mmWave systems are not able to cover all users' locations and only a few users will be supported. Further, frequency band division causes the allocated bandwidth for each user in a dense mmWave network to become small. So,  mmWave networks may not have enough bandwidth to support the users with the required high data-rate.  The obstacles related to using CDMA in mmWave frequencies have been explained in~\cite{wei2018multi}. The propagation characteristics of mmWave frequencies are another reason to incorporate the hybrid beamforming systems and NOMA. Transmission in mmWave band  suffers from high path loss and thus users in different locations may experience very different channel gains.
This implies that mmWave band  better suits  power domain NOMA which offers a larger spectral efficiency when  the channel gain difference between the users is high. Severe shadowing and blockage are other factors that make mmWave links vulnerable to outage~\cite{r1,rappaport2014millimeter,7593259}. Although the large unused spectrum in mmWave bands is envisioned a promising solution for high data-rate transmission in 5G networks, high path loss and outage due to shadowing and blockage make mmWave links prone to temporary shutdowns. Hence, when the link exists, increasing the spectral efficiency will lead to higher data-rate. This would meet the required unprecedented  throughput  of 1000$\times$ current networks in 5G networks.

Integration of NOMA into mmWave systems,   which allows multiple users to share the same beam or the same RF chain,   has been  received considerable research interests~\cite{ding2017random,  ding2017noma,  wang2017spectrum,  hao2017energy,  wei2018multi, xiao2018joint,  wu2017non, 8493528}. In~\cite{ding2017random}, a random beamforming technique is designed for mmWave NOMA systems where  the BS randomly radiates a directional beam toward paired users. In~\cite{ding2017noma},   it is shown that mismatch between the users' channel vector and finite resolution analog beamforming\footnote{Finite resolution analog beamforming is due to the use of a finite number of phase-shifters in the analog beamformer.} simplifies utilizing NOMA in MIMO mmWave systems. In~\cite{wang2017spectrum}, a combination of beamspace MIMO and NOMA is proposed to
ensure that  the number of served users is not limited to the number of RF chains.  In~\cite{hao2017energy},   NOMA is studied for hybrid mmWave MIMO systems,   where a power allocation algorithm has been provided in order to maximize energy efficiency. In all aforementioned works,   NOMA is combined with mmWave systems assuming only baseband precoders/combiners. The works in~\cite{wei2018multi,xiao2018joint,wu2017non,8493528} have recently studied  NOMA in hybrid beamforming systems. Ref.~\cite{wei2018multi} proposes a beam splitting NOMA scheme for hybrid beamforming mmWave systems. In order to increase the spectral efficiency, some users are served with a common RF chain but the grated beams. This technique is only proper when the angle of the directional beams serving the users is large enough. Also, beam grating divides the power of a strong mmWave beam. Hence, far users cannot capture the required power. In~\cite{xiao2018joint}, designing beamforming vectors and allocating power for just two users have been studied. In~\cite{wu2017non},  it is demonstrated that due to the utilization of HB,  the digital precoder of the BS is not perfectly aligned with the user's effective channel. Then, a power allocation algorithm that maximizes the sum-rate has been proposed. Only two users in each beam is considered; moreover, the work fails to study the effect of analog beamforming on the rate performance. Newly, Zhou $et~al.$ have proposed an angle-based user pairing strategy~\cite{8493528}. The strategy repeatedly switches between NOMA and OMA techniques. Such that, when beamwidth of mainlobe of BS is not smaller than the angle difference between two users, they are considered as NOMA users. Otherwise, they are treated as OMA users. Then, the coverage probability and the sum-rate are evaluated. Regularly switching between NOMA and OMA will add more hardware complexity to the system. Also, as it is mentioned, OMA techniques may not be a practical choice for mmWave systems. \textit{In mmWave systems, due to the directional nature of beams in mmWave systems, beam misalignment between the BS and users is inevitable~\cite{6835179}}. Most of the reviewed works  consider neither the effect of phase-shifters employed in the analog beamformer of a HB system nor the effect of beam misalignment.  

In this paper, we investigate the impact of exploiting NOMA in multiuser HB systems termed HB-NOMA. At the outset, it is supposed that HB-NOMA users are paired with respect to their locations and effective channels which is widely adopted by recent research works~\cite{ding2017random,  ding2017noma,  wang2017spectrum,  hao2017energy,  xiao2018joint,  wu2017non, wei2018multi, 8493528}. The achievable rate is evaluated when the BS and users' beam are aligned and misaligned. Essentially, the perfect beam alignment is attributed to the existence of LoS channel aligned in the same direction between the BS and users which allows the users to steer their beam directly toward the BS. The imperfect beam alignment (misalignment) occurs due to practical phenomena such as misaligned LoS channels and NLoS channels which are caused by shadowing and blockage. To the best of authors' knowledge, this paper is the first research work that studies the effect of integration of hybrid beamforming and NOMA on the achievable rate in the presence of beam alignment and misalignment. The contribution of this paper is summarized as follows.
\begin{enumerate}
    \item We incorporate the 5G enabling technology NOMA  and a multiuser HB system studied in~\cite{alkhateeb2015limited}. Since we aim to evaluate the impact of beam misalignment on the downlink of HB-NOMA systems, a sum-rate expression is formulated. Specifically, we revise the sum-rate expression in~\cite{alkhateeb2015limited} with regard to the NOMA technique. Then, an algorithm is introduced to maximize the system sum-rate subject to a total power constraint, in  three steps. To get the first and second steps, we design the analog and digital precoders only regarding LoS channels using the well-known strong effective channel-based effective channel precoder. The third step is a location-based static power allocation.
    \item As the maximized sum-rate directly depends on the effective channels of users, we first study the rate for perfect beam alignment where all users exploit LoS channels. A lower bound is derived for the achievable rate of an HB-NOMA user. The bound reveals that the interference is just due to using NOMA in which SC technique at transmitter and SIC at the receiver are exploited. That is to say, the interference on a user is caused by NOMA users located inside the same cluster called intra-cluster interference. Indeed, HB slightly amplifies the noise term which is led by analog devices used in the beamformer. The analysis shows that for the perfect alignment, the HB-NOMA users can achieve a rate which is close to that of NOMA with the fully-digital beamforming systems.
    \item We study the achievable rate of the maximized sum-rate for misaligned beams between the BS and users in the presence of misaligned LoS and NLoS channels. Toward this goal, the beam misalignment problem is modeled by a beam misalignment factor. Considering the derived factor, the effective channel of the users with misaligned LoS or NLoS channel is described in terms of the aligned effective channel parameter and the misalignment factor.
    \item We extract a lower bound for the achievable rate using the effective channel model. Three terms, i.e., intra-cluster interference, inter-cluster interference,   and noise,   constrain the achievable rate. Unfortunately,   these terms are directly or indirectly associated with misalignment factors. It is concluded that in HB-NOMA with the precoder based on the strongest effective channel the achievable rate of a user depends on both the effective channel gain and beam alignment issue. This is opposite to the fully-digital NOMA systems in which only the effective channel gain affects the rate. Then, an upper bound for rate gap between the aligned and misaligned HB-NOMA user is found.   
\end{enumerate}
To confirm the analyses and the derived expressions,  numerical simulations are done. Different HB-NOMA system parameters are evaluated. The simulations indicate that the HB-NOMA outperforms OMA.

The paper is organized as follows: Section~\ref{sec:system} presents the
system model of HB-NOMA and formulates a sum-rate expression. In Section~\ref{sec:problem},  we maximize the sum-rate for perfect beam alignment then analyze the rate performance. Section~\ref{sec:lower} studies the rate performance for beam misaligned HB-NOMA.
In Section~\ref{sec:simulation},   we present simulation results investigating the rate performance of HB-NOMA. Section~\ref{sec:conclusion} concludes the paper.

\textbf{Notations:} Hereafter,   $j = \sqrt{-1}$,   small letters,   bold letters and bold capital letters will designate scalars,   vectors,   and matrices,   respectively. Superscripts $(\cdot)^{T}$, $(\cdot)^*$ and $(\cdot)^{\dagger}$ denote the transpose, conjugate  and transpose-conjugate operators,   respectively. Further,   $|\cdot|$,   $\norm[]{\cdot}$, and $\norm[]{\cdot}_2$ denote the absolute value, norm-$1$ of $(\cdot)$, and  norm-$2$ of vector $(\cdot)$, respectively. Indeed,   $\norm[]{\cdot}_F$ denotes the Frobenius norm of matrix $(\cdot)$. Finally,   $\mathbb{E}[\cdot]$ denotes the expected value of $(\cdot)$.

\section{System Model and Rate Formulation}\label{sec:system}
\subsection{System Model for HB-NOMA}
We assume a narrow band mmWave downlink system composed of a BS and multiple users as shown in Fig.~\ref{fig:system}. The BS is equipped with  $N_\text{RF}$ chains and $N_\text{BS}$ antennas whereas each user has one RF chain and $N_\text{U}$ antennas. Each RF chain is connected to the antennas through phase-shifters. We also assume that the BS communicates with each user via only one stream. This will be justified later in the present section. In traditional multiuser systems based on the hybrid beamforming the maximum number of users that can be simultaneously served by the BS equals the number of BS RF chains~\cite{alkhateeb2015limited}.
\begin{figure*}
\vspace*{-0.7cm}
\hspace*{-0.8cm}
    \centering
    \includegraphics[scale=1.1]{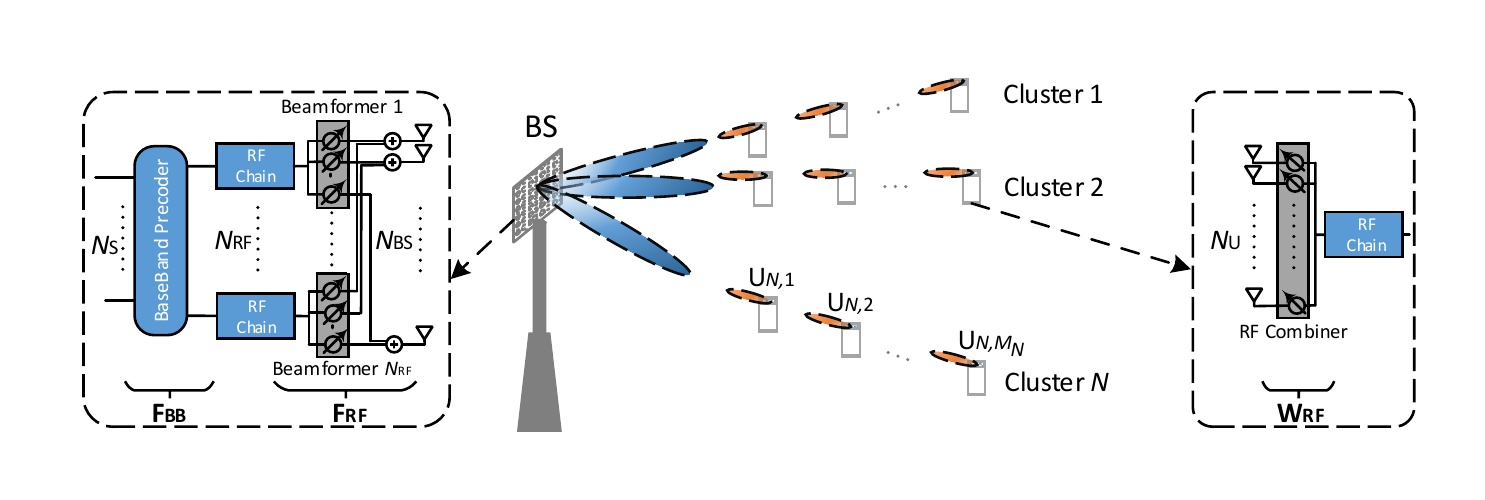}
    \vspace*{-0.5cm}
    \caption{HB-NOMA with one BS and huge number of users grouped into $N$ clusters each with $M_n$ NOMA users.
    $N_\text{S},   N_\text{RF},   N_\text{BS},  $ and $ N_\text{U}$ are the numbers of multiplexed streams,   RF chains,   BS antennas,   and user antennas,   respectively.}
    \label{fig:system}
\end{figure*}

In order to establish a better connectivity in dense areas and further improve the sum-rate, this paper develops HB-NOMA system. The system is practical and takes the parameters of the promising hybrid beamforming into account. To achieve this, we utilize NOMA in hybrid beamforming multiuser systems where  each beam is allowed to serve more than one user. The transmitter simultaneously sends $N_\text{S}$ streams toward $\sum_{n=1}^NM_n$ users which are grouped into $N\leq N_\text{RF}$ clusters. $M_n$ denotes the number of users in the $n$th cluster. The users in each cluster can be scheduled by using the efficient approaches presented in~\cite{6328212,8239622}.  Without loss of generality,   we assume $N_\text{S} = N$. Hence,   $\sum_{n=1}^NM_n \gg N_\text{RF}$; i.e., an HB-NOMA system can simultaneously serve $\sum_{n=1}^NM_n$ users which is much larger than the number of RF chains. In the following we formulate the transmit and received signals for the HB-NOMA system. 

\subsubsection{Superposition coding}
On the downlink of the HB-NOMA system, first, the transmit symbols are superposition coded at the BS. Let $\mathbf{s} = [s_1,   s_2,   \dots,   s_N]^{T}$ denote the information signal vector such that $\mathbb{E}\left[s_ns_n^*\right] = \frac{1}{N}$. Each $s_{n} = \sum_{m = 1}^{M_n}\sqrt{P_{n,  m}}s_{n,  m}$ is the superposition coded signal performed by NOMA with $P_{n,  m}$ and $s_{n,  m}$ being transmit power and transmit information signal for the $m$th user in the $n$th cluster.  Then, the hybrid beamforming is done in two stages. In the first stage, the transmitter applies an $N\times N$ baseband precoder $\mathbf{F}_\text{BB}$ using its $N_\text{RF}$ RF chains. This stage then is followed by an $N_\text{BS} \times N$ RF precoder $\mathbf{F}_\text{RF}$ using analog phase-shifters. Thus, the transmit signal vector after superposition coding is given by
\begin{equation} \label{eq1}
    [x_{1}, x_{2}, \dots, x_N]^T = \mathbf{F}_\text{RF}\mathbf{F}_\text{BB}[s_{1}, s_{2}, \dots, s_N]^T, 
\end{equation}
where $x_n$ denotes the transmit signal toward the $n$th cluster. Hereafter,   U$_{n,  m}$ denotes the $m$th user in the $n$th cluster. Since $\mathbf{F}_\text{RF}$ is implemented by using analog phase-shifters it is assumed that all elements of $\mathbf{F}_\text{RF}$ have an equal norm,   i.e.,   $|\left(\mathbf{F}_\text{RF}\right)_{n,  m}|^2  = N_\text{BS}^{-1}$. Also,   the total power of the hybrid transmitter is limited to $\norm[\big]{\mathbf{F}_\text{RF}\mathbf{F}_\text{BB}}^2_F = N$~\cite{el2014spatially,alkhateeb2015limited}. 

\subsubsection{Successive interference cancellation}
The received signal at  U$_{n,  m}$ is given by
\begin{equation}\label{eq2}
    \mathbf{r}_{n,  m} = \mathbf{H}_{n,  m}\mathbf{F}_\text{RF}\mathbf{F}_\text{BB}\mathbf{s} + \mathbf{n}_{n,  m}, 
\end{equation}
where $\mathbf{H}_{n,  m}$ of size $N_\text{U}\times N_\text{BS}$ denotes the mmWave channel between the BS and U$_{n,  m}$ such that $\mathbb{E}\left[\norm[\big]{\mathbf{H}_{n,  m}}^2_F\right] = N_\text{BS}N_\text{U}$. $\mathbf{n}_{n,  m}\widesim{} \mathcal{CN}(\mathbf{0},  \sigma^2\mathbf{I})$ is the additive white Gaussian noise vector  of size $N_\text{U}\times 1$. Each component of $\mathbf{n}_{n, m}$ has zero-mean and $\sigma^2$ variance. $\mathbf{I}$ denotes the identity matrix of size $N_\text{U}\times N_\text{U}$.
At  U$_{n,  m}$,   the RF combiner  is used to process the received vector as
\begin{align}\label{eq3}
    y_{n,  m} &=  \underbrace{\mathbf{w}_{n,  m}^\dagger\mathbf{H}_{n,  m}\mathbf{F}_\text{RF}\mathbf{f}^n_\text{BB}\sqrt{P_{n,  m}}s_{n,  m}}_{\text{desired signal}}  +\underbrace{\mathbf{w}_{n,  m}^\dagger\mathbf{H}_{n,  m}\mathbf{F}_\text{RF}\mathbf{f}^n_\text{BB}\sum_{k=1, k \neq m}^M\sqrt{P_{n,  k}}s_{n,  k}}_{\text{intra-cluster interference}}  \nonumber \\
    & \ \ + \underbrace{\mathbf{w}_{n,  m}^\dagger\mathbf{H}_{n,  m}\sum_{\ell=1, \ell \neq n }^N\mathbf{F}_\text{RF}\mathbf{f}^\ell_\text{BB}\sum_{q=1}^M\sqrt{P_{\ell,  q}}s_{\ell,  q}}_{\text{inter-cluster interference}} + \underbrace{\mathbf{w}_{n,  m}^\dagger\mathbf{n}_{n,  m}}_{\text{noise}}, 
\end{align}
where $\mathbf{w}_{n,  m} \in \mathbb{C}^{N_\text{U}\times 1}$ denotes the  combiner at U$_{n,  m}$. After combining,   each user decodes the intended signal by using SIC as follows. The first user of each cluster, which has the highest channel gain, is allocated the lowest power and the $M_n$th user, which has the lowest channel gain, is allocated the highest power. At the receiver side,  U$_{n, m}$  decodes the intended signal of U$_{n,  k^\prime}$,  i.e.,  $s_{n, k^\prime}$,  for $k^\prime=m+1,  m+2,  \dots,  M_n$ and subtracts it from the received signal $y_{n, m}$. However,  NOMA treats the intended signal of U$_{n, k}$ for $k=1,  2,  \dots,  m-1$ as intra-cluster interference. In this paper, SIC process is assumed to be ideal. When SIC is non-ideal, the user cannot completely remove the signals of some of U$_{n,k^\prime}$ for $k^\prime = m_1, m+2, \dots, M_n$ which degrades the performance of the system~\cite{vaezi2018non}. The effect of non-ideal SIC on NOMA has recently been studied in~\cite{8125754}. The effect of non-ideal SIC on HB-NOMA will be evaluated in the authors' future work. To this end, the BS should send the order of superposition coding to all users in the cluster. Usually NOMA users are selected to have very different channel gains, specially in mmWave frequencies in which path loss is higher that sub-6 GHz frequencies. So, the order of decoding can be estimated from the user's distance to the BS or its channel gain, correspondingly. We note that the order of encoding is  related to the channel gain as indicated in Section~\ref{algorithm}. 

\subsubsection{Channel model} In mmWave communications, the extended Saleh-Valenzuela model as a multi-path channel (MPC) model has been widely adopted for hybrid beamforming systems~\cite{el2014spatially,alkhateeb2015achievable,alkhateeb2015limited}. In this model, each LoS and NLoS path is described by a channel gain and array steering/response vector at the transmitter/receiver. Here, the number of paths between the BS and U$_{n,m}$ is defined by $A_{n,m}$. The channel matrix is given by  
\begin{equation}\label{eq4}
    \mathbf{H}_{n,  m} = \sqrt{\frac{N_\text{BS}N_\text{U}}{A_{n,m}}} \sum_{\alpha=1}^{A_{n,m}}\beta_{n,  m, \alpha}\mathbf{a}_\text{U}(\vartheta_{n,  m, \alpha}^\text{Az},\vartheta_{n,  m, \alpha}^\text{El})\mathbf{a}_\text{BS}^\dagger(\varphi_{n, m, \alpha}^\text{Az},\varphi_{n, m, \alpha}^\text{El}), 
\end{equation}
where $\beta_{n,  m, \alpha} = g_{n,  m, \alpha}d_{n,  m, \alpha}^{\frac{-\nu}{2}}$ with $g_{n,  m, \alpha}$ is the complex gain with zero-mean and unit-variance for the $\alpha$th MPC,   $d_{n,  m, \alpha}$ is the distance between the BS and U$_{n,  m, \alpha}$,  and $\nu$ is the path loss factor. $\vartheta_{n,  m, \alpha}^\text{Az}$ ($\vartheta_{n,  m, \alpha}^\text{El}$) and $\varphi_{n, m, \alpha}^\text{Az}$ ($\varphi_{n, m, \alpha}^\text{El}$) are normalized azimuth (elevation) angle of arrival (AoA) and angle of departure (AoD), respectively. Also,    $\mathbf{a}_\text{BS}$ and $\mathbf{a}_\text{U}$ are the antenna array steering/response vector of the BS/U$_{n,  m}$. In mmWave outdoor communications, to further reduce the interference, sectorized BSs are likely employed~\cite{5783993}. Mostly, each sector in azimuth domain is much wider than elevation domain~\cite{5783993}.  Reasonably, we assume that the BS separates the clusters in azimuth domain and considers fixed elevation angles. 
Hence, the BS implements only azimuth beamforming and neglects elevation beamforming. In this case, the antenna configuration is a uniform linear array (ULA) and the superscript El is dropped. For a ULA, the steering vector is defined as
\begin{equation}\label{eq5}
\mathbf{a}_{\text{BS}}(\varphi_{n,  m, \alpha}) = \frac{1}{\sqrt{N_\text{BS}}} \left[1,   e^{-j\pi\varphi_{n,  m, \alpha}},  \dots,   e^{-j\pi(N_\text{BS}-1)\varphi_{n,  m, \alpha}}\right]^T. 
\end{equation}
where $\varphi_{n,  m, \alpha} \in [-1,  1]$  is related to the AoD  $\phi \in [-\frac{\pi}{2},  \frac{\pi}{2}]$ as  $\varphi_{n,  m, \alpha} = \frac{2D\text{sin}(\phi)}{\lambda}$~\cite{el2014spatially,  alkhateeb2015limited}. Note that $D$ denotes the antenna spacing and $\lambda$ denotes the wavelength of the propagation. The antenna array response vector for $\mathbf{a}_\text{U}(\vartheta_{n,  m, \alpha})$ can be written in a similar fashion. 

It is mentioned that transmission at mmWave systems is done through directional beams. Since the BS is equipped with HB system, the beamforming can be conducted as follows. When both LoS and NLoS components are available, because LoS component is stronger than NLoS it is reasonable to steer the beam toward LoS component. When only NLoS channels are available, the beam would be steered toward the strongest NLoS component. Thus, only one stream is sent for each cluster. This will also lead to low hardware cost and power consumption due to using one RF chain per stream. Therefore, with a single path component, i.e., $A_{n,m}=1$, the MPC model described in~(\ref{eq4}) is converted to a single path channel given by 
\begin{equation}\label{equ1}
    \mathbf{H}_{n,m} = \sqrt{{N_\text{BS}N_\text{U}}} \beta_{n,  m}\mathbf{a}_\text{U}(\vartheta_{n,  m})\mathbf{a}_\text{BS}^\dagger(\varphi_{n, m}).
\end{equation}
\subsection{Rate Formulation}
In~(\ref{eq3}), after applying superposition coding at the transmitter,   each user experiences two types of interference. Intra-cluster interference which is due to other users within the cluster and inter-cluster interference which is due to users within other clusters. Suppressing the intra-cluster interference directly depends on efficient power allocation and deploying SIC which is discussed in the previous section. To mitigate the inter-cluster interference,   the transmitter needs to design a proper beamforming matrix which will be discussed in Section~\ref{algorithm}. Hence, the rate for  U$_{n,  m}$ is expressed as
\begin{equation}\label{eq6}
    R_{n,  m} = \text{log}_2\left(1 + \frac{P_{n,  m}\left|\mathbf{w}_{n,  m}^\dagger\mathbf{H}_{n,  m}\mathbf{F}_\text{RF}\mathbf{f}_\text{BB}^n\right|^2}{I_\text{intra}^{n,  m} + I_\text{inter}^{n,  m} + \sigma^2}\right), 
\end{equation}
where $I_\text{intra}^{n,  m}$ is given by
\begin{equation}\label{eq61}
I_\text{intra}^{n,  m} =
\sum_{k=1}^{m-1}P_{n,  k}\left|\mathbf{w}_{n,  m}^\dagger\mathbf{H}_{n,  m}\mathbf{F}_\text{RF}\mathbf{f}_\text{BB}^n\right|^2, 
\end{equation}
denotes the intra-cluster. Also,   $I_\text{inter}^{n,  m}$ is defined as
\begin{equation}\label{eq62}
I_\text{inter}^{n,  m} =\displaystyle  \sum_{\ell=1, \ell\neq n}^{N}\sum_{q=1}^{M_n}P_{\ell,  q}\left|\mathbf{w}_{n,  m}^\dagger\mathbf{H}_{n,  m}\mathbf{F}_\text{RF}\mathbf{f}_\text{BB}^\ell\right|^2,
\end{equation}
denotes the inter-cluster interference.

\section{Perfect Beam Alignment: Rate Maximization and Analysis}\label{sec:problem}
\subsection{The Maximization Algorithm}\label{algorithm}
To optimize the sum-rate performance,   hybrid precoder ${\mathbf{F}}_\text{RF}$,   and ${\mathbf{F}}_\text{BB}$,   combiner ${\mathbf{w}}_{n,  m}$ and transmit power ${P}_{n,  m}$ for $m = 1,   2,   \dots,   M_n$ and $n = 1,   2,   \dots,   N$ should be found from
\begin{maxi!}
{\mathbf{F}_\text{RF}, \mathbf{F}_\text{BB}, \mathbf{w}_{n,m}, P_{n,m}}{\sum_{n=1}^N\sum_{m=1}^{M_n} R_{n,  m} \label{eq:objectiveopt}}{\label{eq:opt}}{}
\addConstraint{\left|\left(\mathbf{F}_\text{RF}\right)_{n,  m}\right|^2}{= N_\text{BS}^{-1}\label{b}}
\addConstraint{\norm[\big]{\mathbf{F}_\text{RF}\mathbf{F}_\text{BB}}^2_F}{= N \label{c}}    
\addConstraint{\left|\mathbf{w}_{n,  m}\right|^2}{= N_\text{U}^{-1} \label{d}} 
\addConstraint{\sum_{n=1}^N\sum_{m=1}^{M_n} P_{n,  m}}{\leq P \label{e}}
\addConstraint{P_{n,m}}{>0, \label{f}}
\end{maxi!}
where $P$ equals to the total transmit power. In the above optimization problem, the constraints~(\ref{b}) and (\ref{d}) ensure that all elements of $\mathbf{F}_\text{RF}$ and $\mathbf{w}_n$ have an equal norm. Further, the constraint~(\ref{c}) ensures that the total power of the hybrid transmitter is limited to $N$. The constraint~(\ref{e}) guarantees that the total transmit power is limited to $P$. Finally, (\ref{f}) ensures that the allocated power to U$_{n,m}$ is greater than zero. One would add fairness constrain to the maximization problem.  Ref.~\cite{8125754} discusses a viable solution in this case. In particular, a weighted sum-rate which considers a special priority for each user is utilized. Also, to ensure that all the users achieve a predefined minimum rate $R_\text{min}$, another constrain can be included in the problem~(\ref{eq:opt}) such that $R_{n,m} \ge R_\text{min}$. In this case, an iterative algorithm that properly allocates the power is required~\cite{zhang2016robust}. Without loss of generality, here, we assume that all the users satisfy $R_{n, m}\ge R_\text{min}$.      

It is mentioned that transmission in mmWave bands happens through LoS and NLoS channels. In particular, the users which are located far from the BS will mostly be supported via NLoS channels~\cite{7593259}.  Let first focus on only LoS channels. We assume that all channels are LoS and the effective channels are perfectly aligned as shown in Fig.~\ref{fig:system}. By perfect alignment we mean that $\mathbf{a}_\text{BS}(\varphi_{n,  m})$ is identical for all users in the $n$th cluster,   i.e.,    $\mathbf{a}_\text{BS}(\varphi_{n,  1}) = \mathbf{a}_\text{BS}(\varphi_{n,  2}) = \dots = \mathbf{a}_\text{BS}(\varphi_{n,  M_n})$ for $n = 1,   2,  \dots,  N$.

In general, there are two extreme cases to design baseband precoder for mmWave-NOMA systems, strong effective channel-based and singular value decomposition (SVD)-based precoder methods~\cite{wang2017spectrum}. The strong effective channel-based is designed for only LoS channels and the SVD-based precoder is designed for only NLoS channels. Further, to the best of authors' knowledge, it is not shown how to design the SVD-based RF precoder for hybrid beamforming system. Here, in order to understand the behavior of beam misalignment in HB-NOMA systems we choose the strong effective channel-based precoder which is widely used in the literature~\cite{wang2017spectrum,hao2017energy,wu2017non}. 

The maximization problem in~(\ref{eq:opt}) is non-convex and finding the optimal solution is not trivial.
To ease, we present an efficient and simple algorithm in three steps as described below.

In the first step,   the BS and U$_{n,  m}$ solve the following problem
\begin{align}\label{eq7}
    \underset{\mathbf{w}_{n,  m},  \mathbf{f}_\text{RF}^{n,  m}}{\text{maximize}} \quad \left|\mathbf{w}_{n,  m}^\dagger\mathbf{H}_{n,  m}\mathbf{f}_\text{RF}^{n,  m}\right| \qquad
    \text{subject to (\ref{b}) and (\ref{d})}.
\end{align}
Since the channel $\mathbf{H}_{n,  m}$ has only one path,   and given the continuous beamsteering capability assumption,   in view of \eqref{eq4}, $\mathbf{w}_{n,  m}=\mathbf{a}_\text{U}(\vartheta_{n, m})$ and ${\mathbf{f}}_\text{RF}^{n,  m} = \mathbf{a}_\text{BS}(\varphi_{n,  m}),$ are  the optimal solutions~\cite{alkhateeb2015limited}. We design the RF (analog) and baseband (digital) precoders using the adopted strong effective channel-based method. Hence, in order to design the RF precoder, the BS selects the first user of each cluster. The RF precoder of the first user of the $n$th cluster makes the $n$th column of  the RF precoding matrix,   i.e.,    ${\mathbf{f}}_\text{RF}^{n,  1}$, gives the RF precoding matrix as
 \begin{equation}\label{eq81}
 \mathbf{F}_\text{RF} = \left[{\mathbf{f}}^{1,  1}_\text{RF},   {\mathbf{f}}^{2,  1}_\text{RF},  \dots,  {\mathbf{f}}^{N,  1}_\text{RF}\right].
 \end{equation}
The first user is determined based on the locations of the user as follows:
\begin{equation}\label{eq8}
    \left|\beta_{n,  1}\right| \geq \left|\beta_{n,  2}\right| \geq \dots \geq \left|\beta_{n,  M_n}\right|,   \quad \text{for} \quad n = 1,   2,   \dots,   N, 
\end{equation}
where $\beta_{n,  m}$ is the gain factor defined in~(\ref{eq4}). To determine the first user, the BS does not need to know the channel gain of the users. Recall that the channel gain $\beta_{n,m}$, defined in~(\ref{eq4}), mainly depends on distance between the BS and U$_{n,m}$ ($d$) and path loss factor ($\nu$). Since the path loss factor is identical for all users, the first user of each cluster can be determined as the closest user to the BS such that its channel gain has the highest amplitude among the users in the same cluster. While the purpose of ordering in~(\ref{eq8}) is to define the first user, to realize NOMA, another ordering method based on the effective channel gain is presented in the third step. It should be stressed that the main reason to design the digital precoder with respect to the strongest channel is that the strongest user must decode the other users' signal before its signal. So, the power of this user's signal is not affected by other clusters' signal. More details will be provided in Section~\ref{sec:lower}.

In the second step,   the effective channel for  U$_{n,  m}$ is expressed as
\begin{align}\label{eq9}
    \overbar{\mathbf{h}}_{n,  m}^\dagger &=  \mathbf{w}_{n,  m}^\dagger\mathbf{H}_{n,  m}\mathbf{F}_\text{RF}= \sqrt{N_\text{BS}N_\text{U}}\beta_{n,  m}\mathbf{a}_\text{BS}^\dagger(\varphi_{n,  m})\mathbf{F}_\text{RF}.
\end{align}
Regarding the strongest channel-based method,  we write the effective channel matrix as
\begin{equation}\label{eq91}
\overbar{\mathbf{H}} = \left[\overbar{\mathbf{h}}_{1,  1},   \overbar{\mathbf{h}}_{2,  1},   \dots,   \overbar{\mathbf{h}}_{N,  1} \right]^\dagger, 
\end{equation}
where $\overbar{\mathbf{h}}_{n,  1}$ denotes the effective channel vector of U$_{n,  1}$.

Designing a proper digital precoder $\mathbf{F}_\text{BB}$ can reduce the inter-cluster interference. In brief,   designing the baseband precoder becomes equivalent to solving
\begin{equation}\label{eq10}
    \underset{\{\mathbf{f}_\text{BB}^\ell\}_{\ell\neq n}}{\text{minimize}} \ I_\text{inter}^{n,  m} \qquad \text{subject to (\ref{c})}.
\end{equation}
where $I_\text{inter}^{n,  m}$ is defined in~(\ref{eq61}). We notice that so far we have designed the analog beamformer and combiner. The only unknown parameter is the digital beamformer.  In this paper, we adopt zero-forcing beamforming (ZFBF) which makes a balance between implementation complexity and performance~\cite{spencer2004zero ,yoo2006optimality}. Based on ZFBF, the solution for (\ref{eq10}) is obtained as~\cite{alkhateeb2015limited}
\begin{equation}\label{eq11}
    \mathbf{F}_\text{BB} = \overbar{\mathbf{H}}^\dagger\left(\overbar{\mathbf{H}}\overbar{\mathbf{H}}^\dagger\right)^{-1}\bf{\Gamma}, 
\end{equation}
where the diagonal elements of $\mathbf{\Gamma}$ are given by~\cite{alkhateeb2015limited}
\begin{equation}\label{eq12}
    \mathbf{\Gamma}_{n,  n} = \sqrt{\frac{N_\text{BS}N_\text{U}}{\left(\mathbf{F}^{-1}\right)_{n,  n}}}\left|\beta_{n,  1}\right|,   \quad \text{for} \quad n = 1,   2,   \dots,   N.
\end{equation}
where $\mathbf{F}=\mathbf{F}_\text{RF}^\dagger\mathbf{F}_\text{RF}$. The determined precoder in~(\ref{eq11}) indicates that inter-cluster interference on first users is zero,   i.e.,   $\overbar{\mathbf{h}}^\dagger_{n,  1}\mathbf{f}^\ell_\text{BB} = 0$ for $n = 1,   2,   \dots,   N$ and $\ell \neq n$. That is,   inter-cluster interference is perfectly eliminated on the first users. This completes our justification about the orienting the beams toward the first users and choosing their effective channel vector in designing $\mathbf{F}_\text{BB}$.

In the third step,   the BS first reorders the users then allocates the power. The reordering process is done based on the effective channel vectors as
\begin{equation}\label{eq121}
    \norm[\big]{\overbar{\mathbf{h}}_{n,  1}} \geq \norm[\big]{\overbar{\mathbf{h}}_{n,  2}}\geq \dots
    \geq \norm[\big]{\overbar{\mathbf{h}}_{n,  M_n}},  \quad \text{for} \quad n = 1,   2,   \dots,   N.
\end{equation}
Notice that in (\ref{eq8}) we aimed to find the first users based on the large-scale gain. However,   in HB-NOMA the power allocation is conducted based on order of the effective channel gains. It is not irrational to assume that the BS knows the effective channels. This can be done through the channel quality indicator (CQI) messages~\cite{chen2017exploiting}. Each user feeds the effective channel back to the BS then it sorts the users.

The optimal power allocation in~(\ref{eq:opt}) can be done by solving the following problem.
\begin{align}\label{eq13}
    \underset{P_{n,  m}}{\text{maximize}} \ \sum_{n=1}^{N}\sum_{m=1}^{M_n} R_{n,  m}  \qquad \text{subject to (\ref{e}) and (\ref{f})}.
\end{align}
To solve the problem,  we propose a two-stage solution. First the BS divides the power between the clusters considering their users' channel gain as follows.
\begin{equation}\label{equPowerAllocation}
    P_n = \frac{\displaystyle\sum_{m=1}^{M_n}\norm[\big]{\overbar{\mathbf{h}}_{n,  m}}^2}{\displaystyle\sum_{n=1}^N\sum_{m=1}^{M_n}\norm[\big]{\overbar{\mathbf{h}}_{n, m}}^2}P, \quad \text{for} \quad n=1, 2, \dots, N.
\end{equation}
Then a fixed power allocation is utilized for the users in each cluster respecting the constraint $\sum_{m=1}^{M_n} P_{n,  m} = P_n$. To determine $P_{n,m}$, one solution is to allocate a certain amount of power for each U$_{n,m}$ except the first one that only satisfies $R_{n,m}=R_\text{min}$, then the remaining is assigned to U$_{n,1}$. This power allocation process is in consist with the concept of NOMA in which, to achieve higher sum-rate, the stronger user should receive more power~\cite{saito2013system,  saito2013non,  ding2014performance,  higuchi2015non}. On the other hand, recall that mmWave channels are vulnerable to blockage and shadowing. Especially, for the weak users which are located far from the BS, this issue becomes worse. So, the weak users may not be able to achieve the required minimum rate. Another solution is to give priority to the fairness issue. To this, we need to allocate less power to the strong users and more power to the weak users. It turns out, fairness works against achieving maximum rate. Thus, our solution to achieve maximum rate and compensate for the mmWave propagation issues is to assign the same amount of power for all the users, i.e.,
\begin{equation}
    P_{n,1} = P_{n,2} = \cdots = P_{n,M_n}.
\end{equation}

\subsection{The Achievable Rate Analysis}
In this section,  the achievable rate of U$_{n,m}$ is evaluated with respect to the designed parameters. We derive a lower bound which characterizes insightful results on the achievable rate of HB-NOMA.

\begin{theorem}\label{theo:1}
\normalfont
With perfect beam alignment, a lower bound on  the achievable rate of  U$_{n,  m}$ is given by
\begin{equation}\label{eq14}
    \overbar{R}_{n,  m} \geq \text{log}_2\left(1 + \frac{P_{n,  m}N_\text{BS}N_\text{U}\left|\beta_{n,  m}\right|^2}{\displaystyle\sum_{k=1}^{m-1}P_{n,  k}N_\text{BS}N_\text{U}\left|\beta_{n,  m}\right|^2 + \sigma^2 \kappa_\text{min}^{-1}(\mathbf{F})}\right), 
\end{equation}
$\kappa_\text{min}(\mathbf{F})$ denotes the minimum eigenvalue of $\mathbf{F}$.
\end{theorem}
\begin{proof}
Please see Appendix~\ref{app:Theorem1}.
\end{proof}

\begin{remark}\label{remark:1}
\normalfont
Theorem 1 indicates that when the alignment between the users in each cluster is perfect,   still two terms degrade the sum-rate performance of every HB-NOMA user. The first term $\sum_{k = 1}^{m-1}P_{n,  k}N_\text{BS}N_\text{U}\left|\beta_{n,  m}\right|^2 $ is due to using NOMA scheme which leads to inevitable intra-cluster interference. The second term $\kappa_\text{min}^{-1}(\mathbf{F})$ is due to realizing the beamforming with digital and analog components,   i.e.,   hybrid beamforming instead of fully-digital components.  It is worth mentioning that in the fully-digital beamforming the first term exists but the second term is always one. Therefore,   even under perfect beam alignment assumption the hybrid beamforming intrinsically imposes small loss on the achievable rate.
\end{remark}
\section{Beam Misalignment: Modeling, Rate Analysis, and Rate Gap}\label{sec:lower}
\begin{figure}
     \centering
     \includegraphics[scale=1.3]{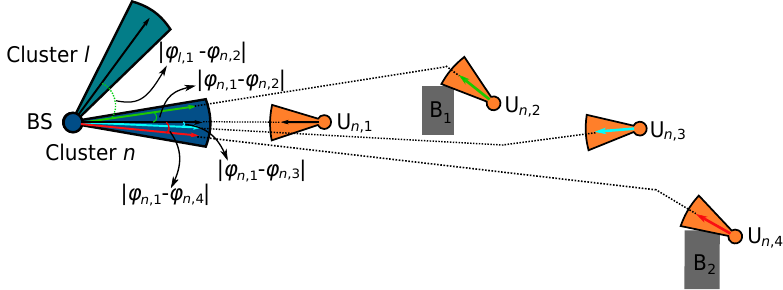}
     \caption{Beam misalignment in mmWave communications due to the NLoS channels. The NLoS channels are caused by blockages B1 and B2.}
     \label{fig:misalignment}
 \end{figure}
In the previous section we designed the precoders when only LoS channels exist and the users are perfectly aligned. The precoders are found based on the strongest effective channel. Perfect alignment is an ideal assumption. In fact, AoDs/AoAs are random variable and with almost surely the probability of occurring different AoDs/AoAs even in LoS channels is one which leads to $\mathbf{a}_\text{BS}(\varphi_{n,  1}) \neq \mathbf{a}_\text{BS}(\varphi_{n,  2}) \neq \dots \neq \mathbf{a}_\text{BS}(\varphi_{n,  M_n})$ for $n = 1,   2,  \dots,  N$. On the other hand, recall that in mmWave frequencies, due to shadowing and blockage, NLoS channels are inevitable~\cite{7593259}. These channels force the users to indirectly steer their beam toward the BS as illustrated by Fig.~\ref{fig:misalignment}. So, the misalignment between the effective channel of the first user and the users with misaligned LoS and NLoS channel in each cluster causes the digital baseband precoder cannot eliminate the inter-cluster interference.  As a result, the achievable rate is degraded. In this section, first the misalignment is modeled. Second, using the derived model, a lower bound is found for the rate. Finally, an upper bound is extracted for the rate gap between the perfect alignment and misalignment.    

\begin{remark}
\normalfont
While our findings in this section are general and hold for misaligned LoS  and NLoS channels, we only concentrate on NLoS channels. Thus, by LoS channel we mean a perfectly aligned channel. Also, it is assumed that all users expect the first one in all clusters have NLoS channels.  In order to distinguish effective channel of the users with aligned LoS channels from NLoS channels, hereafter,  we denote $\overbar{\mathbf{h}}_{n, m}$ as effective channel of the user with perfect beam alignment and $\tilde{\mathbf{h}}_{n, m}$ as effective channel of the user with imperfect beam alignment. Also,  $\overbar{R}_{n, m}$ and  $\tilde{R}_{n, m}$ denote the rate of U$_{n, m}$ with LoS and NLoS channel,  respectively. 
\end{remark}

\subsection{Beam Misalignment Modeling}
In what follows,   we study the impact of imperfect beam alignment on the rate. Before that,   we calculate the norm of the effective channel defined in~(\ref{eq9}). Defining
 \begin{equation}\label{eqFejer}
 \left|\mathbf{a}_\text{BS}^\dagger(\varphi_{n,  m})\mathbf{a}_\text{BS}(\varphi_{\ell,  1})\right|^2 = K_{N_\text{BS}}(\varphi_{\ell,  1}-\varphi_{n,  m}), 
 \end{equation}
 where $K_{N_\text{BS}}$ is Fej$\acute{\text{e}}$r kernel of order $N_\text{BS}$~\cite{strichartz2000way},   we get
\begin{equation}\label{eq1601}
\norm[\big]{\tilde{\mathbf{h}}_{n,  m}}^2 = N_\text{BS}N_\text{U}\left|\beta_{n,  m}\right|^2\displaystyle \sum_{\ell = 1}^N K_{N_\text{BS}}\left(\varphi_{\ell,  1}-\varphi_{n,  m}\right).
 \end{equation}
Now, we model the correlation between the effective channels for U$_{n,  m}$ and U$_{n,  1}$ and between U$_{n,  m}$ and U$_{\ell,  1}$ with $\ell\neq n$ by defining them as intra-cluster misalignment factor and inter-cluster misalignment factor, respectively. Notice that we consider the worst scenario. That is, U$_{n,  m}$ for $m=2, 3, \dots, M_n$ receives the signal through NLoS channel, while only U$_{n,  1}$ for $n=1, 2, \dots, N$ receives through LoS channel. Assuming LoS channel for the first users is reasonable, since in mmWave communications the users close to the BS experience LoS channels with high probability~\cite{7593259}. 
\begin{lemma}\label{lemma:2}
\normalfont
The misalignment effective channel of U$_{n,  m}$  and U$_{n,  1}$ can be modeled as
 \begin{equation} \label{eq19}
    \hat{\tilde{\mathbf{h}}}_{n,  m} = \rho_{n,  m}\hat{\tilde{\mathbf{h}}}_{n,  1} + \sqrt{1 - \rho_{n,  m}^2}\hat{\mathbf{g}}^{-n}_\text{BS}, 
\end{equation}
where $\hat{\tilde{\mathbf{h}}}_{n,  m}$ denotes the normalized imperfect effective channel, $\rho_{n,m}$ denotes the misalignment factor obtained as   
\begin{equation} \label{eqrho19}
    \rho_{n,m} =
    \frac{\displaystyle\sum_{i=1}^N\kappa_i(\mathbf{F})\left|\mathbf{a}_\text{BS}^\dagger(\varphi_{n,  m})\mathbf{v}_1^i\mathbf{v}_1^{i\dagger}\mathbf{a}_\text{BS}(\varphi_{n, 1})\right|}{\sqrt{\displaystyle\sum_{\ell = 1}^N K_{N_\text{BS}}\left(\varphi_{\ell,  1}-\varphi_{n,m}\right)}\sqrt{\displaystyle\sum_{\ell = 1}^N K_{N_\text{BS}}\left(\varphi_{\ell,  1}-\varphi_{n,1}\right)}}, 
\end{equation}
where  $\kappa_i(\mathbf{F})$ is the $i$th eigenvalue of $\mathbf{F}$.  $\hat{\mathbf{g}}^{-n}_\text{BS}$ is a normalized vector located in the subspace generated by linear combination of ${{\mathbf{a}}}_\text{BS}(\varphi_{\ell,1})$ for $\ell \neq n$, such that $\hat{\mathbf{g}}^{-n}_\text{BS}=\frac{\mathbf{g}^{-n}_\text{BS}}{\norm[\big]{\mathbf{g}^{-n}_\text{BS}}},$ where $\mathbf{g}^{-n}_\text{BS} = \sqrt{N_\text{BS}N_\text{U}}\mathbf{F}_\text{RF}^\dagger\sum_{\ell=1,\ell\neq n}^N\beta_{\ell,1}\mathbf{a}_\text{BS}(\varphi_{\ell,1}).
$
\end{lemma}
\begin{proof}
Please see Appendix~\ref{app:lemma2}.
 \end{proof}
\subsection{Rate Analysis}
Now we are ready to find a lower bound for the achievable rate of  U$_{n,  m}$.
\begin{theorem}\label{theo:2}
\normalfont
With imperfect beam alignment,   a lower bound on the achievable rate of U$_{n,  m}$, is given by
\begin{equation}\label{eq16}
    \tilde R_{n,  m} \geq \text{log}_2\left(1 +\frac{P_{n,  m}\rho_{n,  m}^2N_\text{BS}N_\text{U}\left|\beta_{n,  m}\right|^2}{\zeta_\text{intra}^{n,  m}+ \zeta_\text{inter}^{n,  m} + \zeta_\text{noise}^{n,  m}}\right), 
\end{equation}
where $\zeta_\text{intra}^{n,  m} = \sum_{k = 1}^{m-1}P_{n,  k}\rho_{n,  m}^2N_\text{BS}N_\text{U}\left|\beta_{n,  m}\right|^2$
and $\zeta_\text{inter}^{n,  m} = \left(1-\rho_{n,m}^2\right)N_\text{BS}N_\text{U}\left|\beta_{n,  m}\right|^2\kappa_\text{max}\left(\mathbf{S}\right)
\kappa_\text{min}^{-1}(\mathbf{F}) \times K_{N_\text{BS},  1}$ in which $\kappa_\text{max}\left(\mathbf{S}\right)$ is the maximum eigenvalue of $\mathbf{S} = \mathbf{F}_\text{BB}^{-n,W}\mathbf{F}_\text{BB}^{-n,W\dagger}$, $\mathbf{F}_\text{BB}^{-n,W}$ denotes the wieghted $\mathbf{F}_\text{BB}$ after eliminating the $n$th column where the columns are scaled by $P_\ell~\forall\ell\neq n$. Also, for some $m$ we define
\begin{equation}\label{eq163}
K_{N_\text{BS},  m} = \displaystyle \sum_{\ell = 1}^N K_{N_\text{BS}}\left(\varphi_{\ell,  1}-\varphi_{n,  m}\right), 
\end{equation}
where $K_{N_\text{BS}}\left(\varphi_{\ell,  1}-\varphi_{n,  m}\right)$ denotes the Fej$\acute{\text{e}}$r kernel in~(\ref{eqFejer}). Finally, $\zeta_\text{noise}^{n,  m}$ is expressed as $\zeta_\text{noise}^{n,  m} = \sigma^2\kappa_\text{min}^{-1}(\mathbf{F})K_{N_\text{BS},  1}K^{-1}_{N_\text{BS},  m},$
where $K_{N_\text{BS},  m}$ is defined in~(\ref{eq163}).
 \end{theorem}
\begin{proof}
Please see Appendix~\ref{app:theorem2}.
\end{proof}

\begin{remark}\label{remark:2}
\normalfont
Since for U$_{n,  1}$ the factor $\rho_{n, 1}$ is one,  we have $\overbar{\mathbf{h}}_{n,  1} = \tilde{\mathbf{h}}_{n,  1}$. Thus,  Theorem~\ref{theo:1} is still valid for these users. 
\end{remark}

\begin{remark}\label{remark:3}
\normalfont
Theorem~\ref{theo:2} states that the achievable rate of each user depends on the intra-cluster and inter-cluster misalignment factors, and a weak alignment reduces the power of the effective channel of that user. Intra-cluster and inter-cluster power allocation are other parameters that affect the achievable rate as seen in~(\ref{eq16}). Further,  the bound shows that the maximum eigenvalue of the baseband precoder is important in maximizing the achievable rate. That is to say, the effective channel matrix should be designed in a way that the eigenvalues of the baseband precoder are as close as possible to each other. This is because if eigenvalues are far from each other,   the maximum eigenvalue will be large. This increases the value of $\zeta_\text{inter}^{n,  m}$ which causes less achievable rate. 
\end{remark}

To gain some insight into the effect of beam misalignment, we extract a lower bound for the rate gap when U$_{n, m}$ receives the signal via LoS and NLoS channel.      
\begin{theorem}\label{theo:3}
\normalfont
The rate gap between the perfect aligned and misaligned U$_{n,m}$ is given by
\begin{align}
    \Delta R_{n,m} &\overset{\Delta}{=}  \overbar R_{n,  m} - \tilde R_{n,  m} \nonumber \\  
    &\leq \text{log}_2\left(1 + \frac{\displaystyle \left(1-\rho_{n,  m}^2\right)\kappa_\text{max}\left(\mathbf{S}\right)+\sigma^{2}K^{-1}_{N_\text{BS},m}N_\text{BS}^{-1}N_\text{U}^{-1}\left|\beta_{n,m}\right|^{-2}}{\rho_{n,  m}^2 K^{-1}_{N_\text{BS},1}\kappa_\text{min}(\mathbf{F})\displaystyle   \sum_{k=1}^{m-1}P_{n,  k}}\right).
\end{align}
\end{theorem}
\begin{proof}
Please see Appendix~\ref{app:theorem3}.
\end{proof}
The upper bound in  Theorem~\ref{theo:3} explicitly shows the effect of the parameters of HB-NOMA system on the rate performance. A low misalignment factor can substantially increase the rate gap.       

\begin{remark}
\normalfont
In Section~\ref{algorithm} the users are assumed to have LoS channels and to be perfectly aligned in a same direction. Particularly, Eq.~(\ref{eq121}) orders the users with respect to the their effective channel. Actually, these effective channels are the strongest path between the BS and users. However, when the users are not aligned in the same direction, the effective channels are not necessarily the strongest. This is because the users have to orient their antenna array response vector toward the beam direction of the first user rather than the best direction. Hence, to properly perform SIC, we revise the ordering considering the misalignment effective channel, i.e.,
\begin{equation}
     \norm[\big]{\tilde{\mathbf{h}}_{n,  1}} \geq \norm[\big]{\tilde{\mathbf{h}}_{n,  2}}\geq \dots
    \geq \norm[\big]{\tilde{\mathbf{h}}_{n,  M_n}},  \quad \text{for} \quad n = 1,   2,   \dots,   N.
\end{equation}
Further, in~(\ref{equPowerAllocation}) the aligned effective channel should be replaced by the misaligned effective channel.  
\end{remark}
\section{Numerical Results}\label{sec:simulation}
In this section we simulate the HB-NOMA system regarding the various design parameters to confirm the analytical derivations in Theorems~\ref{theo:1}-\ref{theo:3}. For simulations, since large scaling fading and path loss put more restriction on mmWave systems, the small scale fading is negligible. The defualt number of antennas $N_\text{BS}$ $N_\text{MU}$ for the BS and all users is assumed 32 and 8, respectively, unless it is mentioned. The misalignment is described as a random variable uniformly distributed by parameter $b$, i.e., $\varphi_{n,1}-\varphi_{n,m} \in [-b, b]$. We first present the results of the HB-NOMA with perfect alignment. Then, the effect of misalignment on the rate performance is shown. Finally, the sum-rate of HB-NOMA with OMA is illustrated.      
\subsection{Perfect Beam Alignment}
Figure~\ref{fig:perfect} studies the performance of the derived bound in Theorem~\ref{theo:1} for aligned users. The users are not affected by the inter-cluster interference from other clusters. It is supposed that the number of users is two and channel gain of the strong and weak user is 0 and -2 dB, respectively. Fig.~\ref{fig:perfect}(a) reveals that the HB-NOMA approximately achieves the rate the same as that of fully-digital beamforming (FD beamforming) for a wide range of SNR. In particular, a small gap between the exact value of HB-NOMA and the lower bound is observed for the strong user (U$_{1,1}$). This is because the complicated expression of the noise term in~(\ref{eq14}) is replaced by a simple but greater term.  For the weak user (U$_{1,2}$) the bound is very tight due to two reasons. First, in the SINR of the weak user, the noise term is dominated by the interference term. Therefore, the effect of noise term is neglected. Second, the interference term is modeled very accurately.\\
Fig.~\ref{fig:perfect}(b) studies the achievable rate for various $N_\text{BS}$. For small $N_\text{BS}$s, the fully-digital outperforms the HB-NOMA. When $N_\text{BS}$ is samll, the RF precoder is not able to steer a highly direct beam toward the users. By increasing $N_\text{BS}$, the beam becomes narrow and the users capture much more power. Again, for the weak user, the lower bound is accurate at all $N_\text{BS}$ regions. For the strong user, the bound does not approach to the exact value but, for $N_\text{BS}>60$, the bound is approximately the same as to the exact HB-NOMA.

 \begin{figure}
     \centering
    \begin{subfigure}[t]{0.45\textwidth}
    \centering
        \raisebox{-\height}{\includegraphics[scale=0.6]{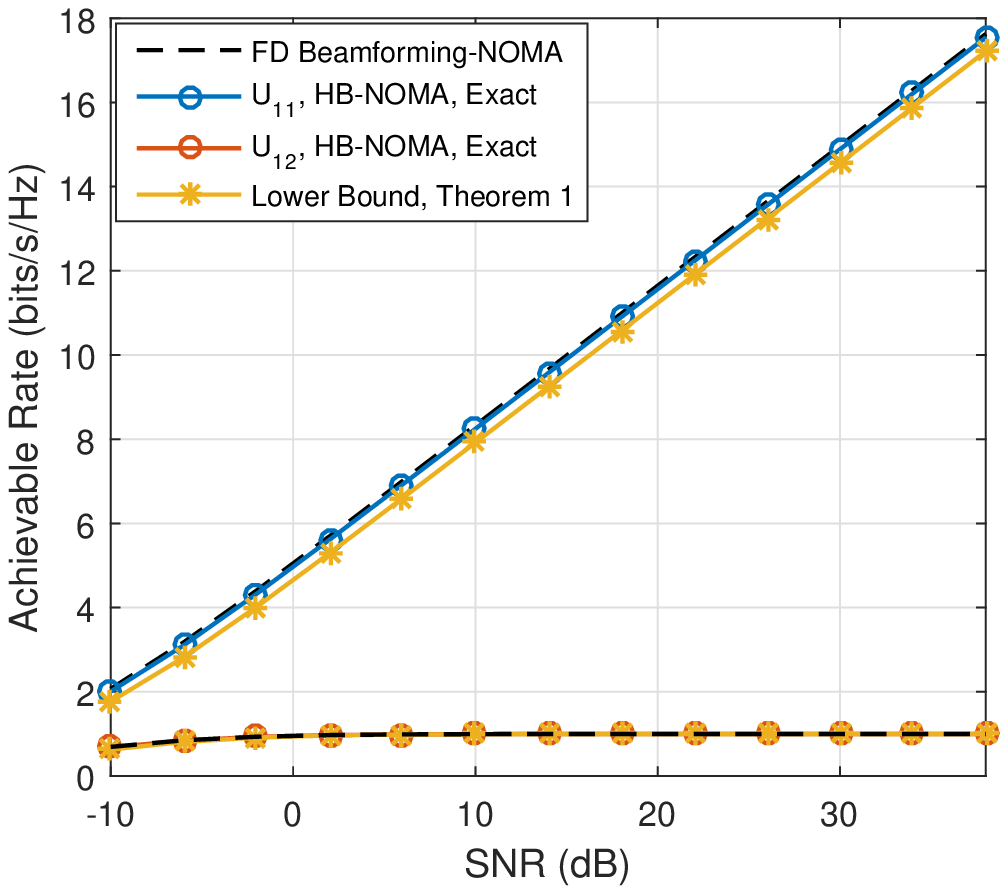}}
        \caption{}
    \end{subfigure}
    \hfill
    \begin{subfigure}[t]{0.45\textwidth}
        \raisebox{-\height}{\includegraphics[scale=0.6]{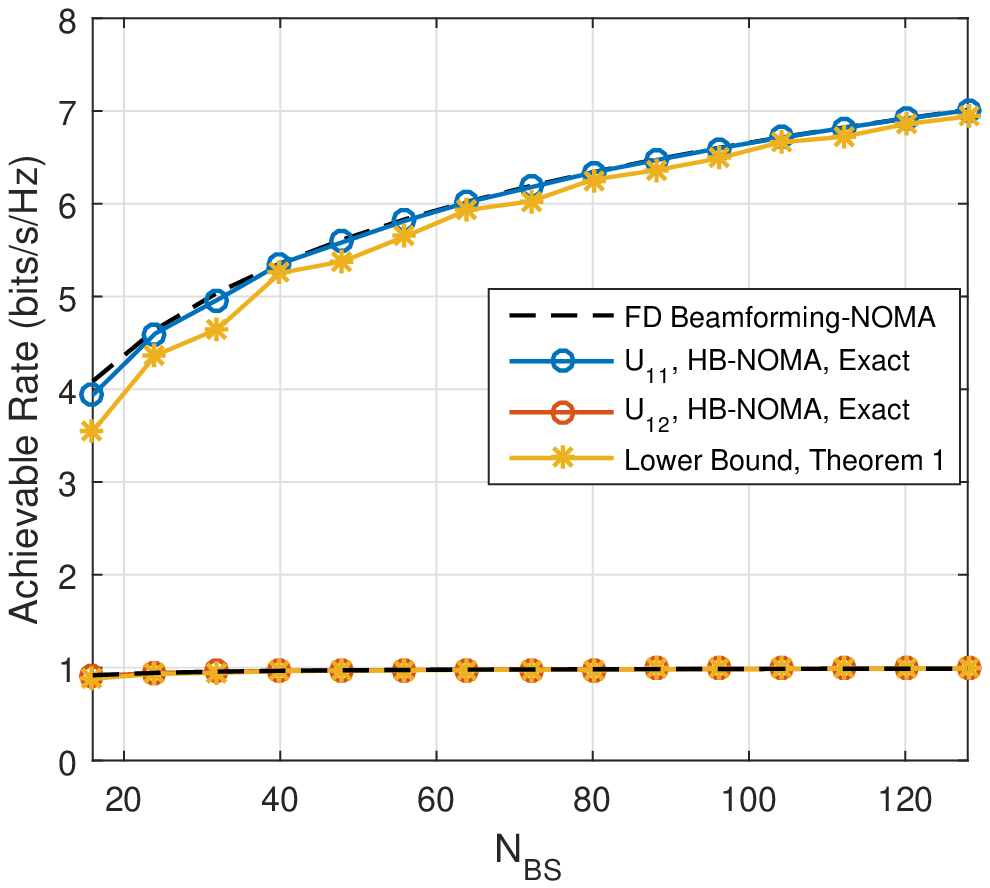}}
         \caption{}
    \end{subfigure}
    \caption{Evaluation of rate performance of the strong channel-based precoder in HB-NOMA with perfect alignment (LoS channels) in terms of (a) SNR and (b) $N_\text{BS}$.}
    \label{fig:perfect}
\end{figure}
 
\subsection{Beam Misalignment}
 
\begin{figure}
     \centering
    \begin{subfigure}[t]{0.45\textwidth}
        \raisebox{-\height}{\includegraphics[scale=0.6]{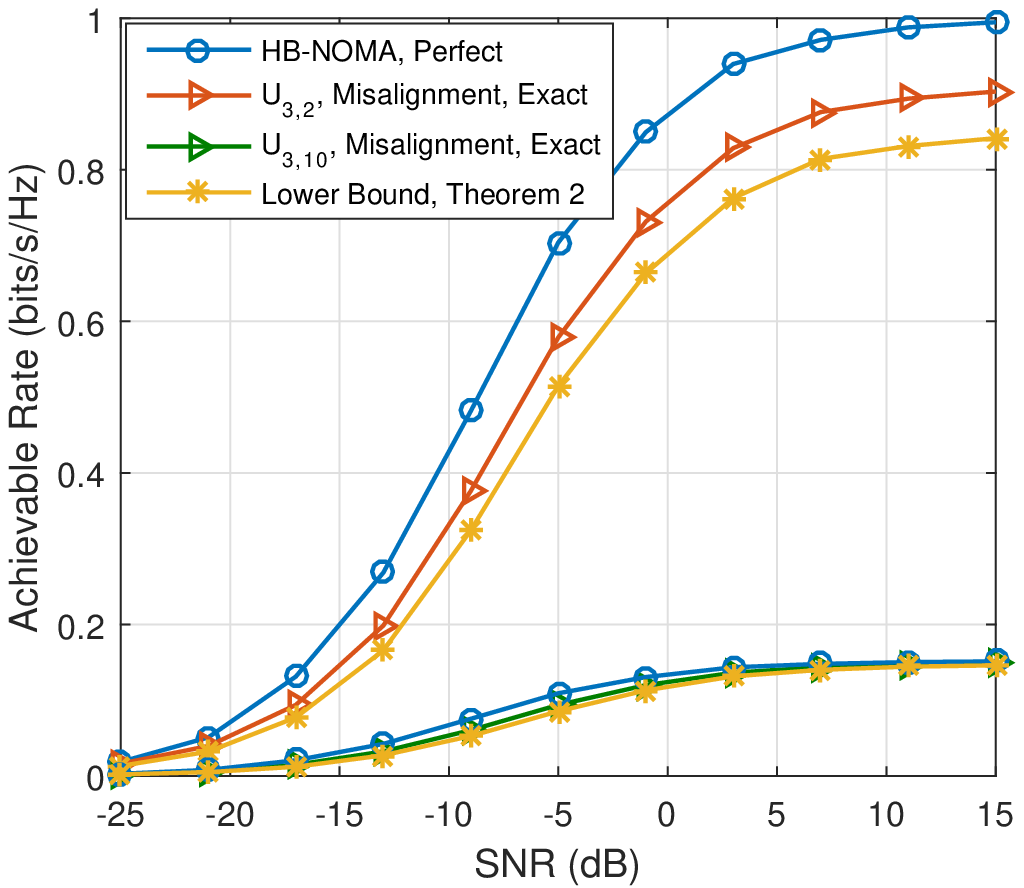}}
        \caption{}
    \end{subfigure}
    \hfill
    \begin{subfigure}[t]{0.45\textwidth}
        \raisebox{-\height}{\includegraphics[scale=0.6]{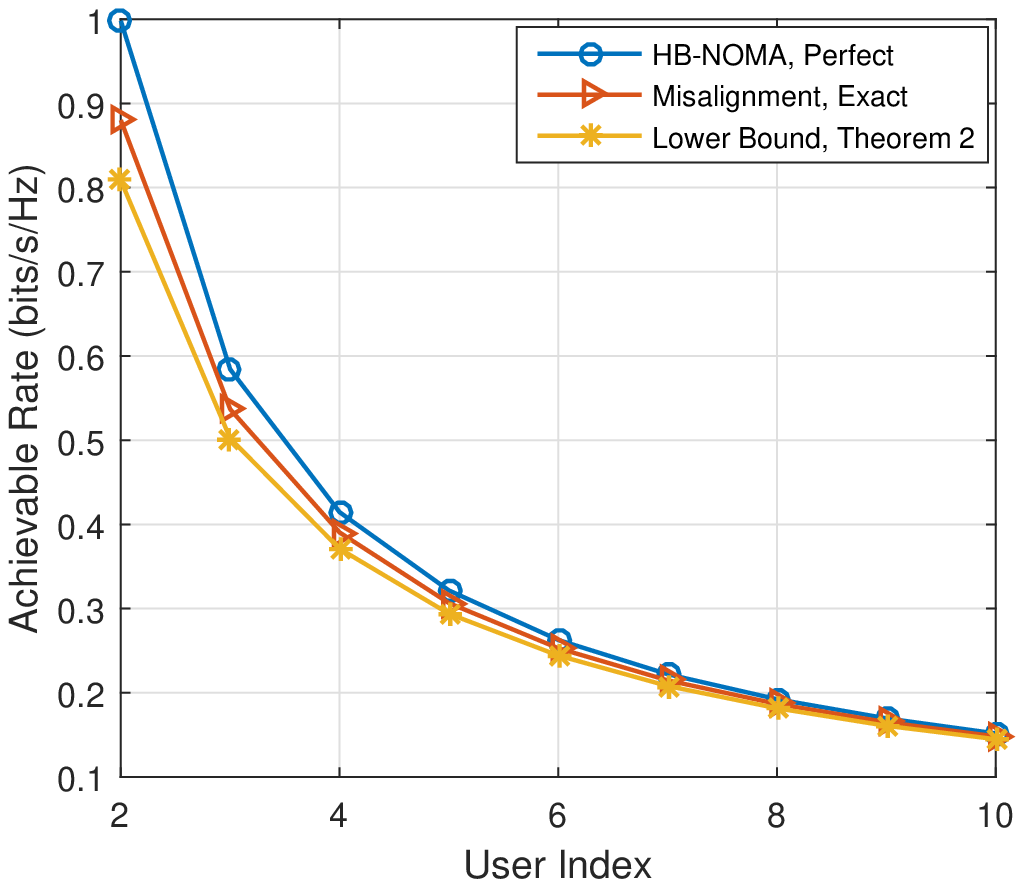}}
         \caption{}
    \end{subfigure}
   
 
   \begin{subfigure}[t]{.45\textwidth}
        \centering
        \raisebox{-\height}{\includegraphics[scale=0.6]{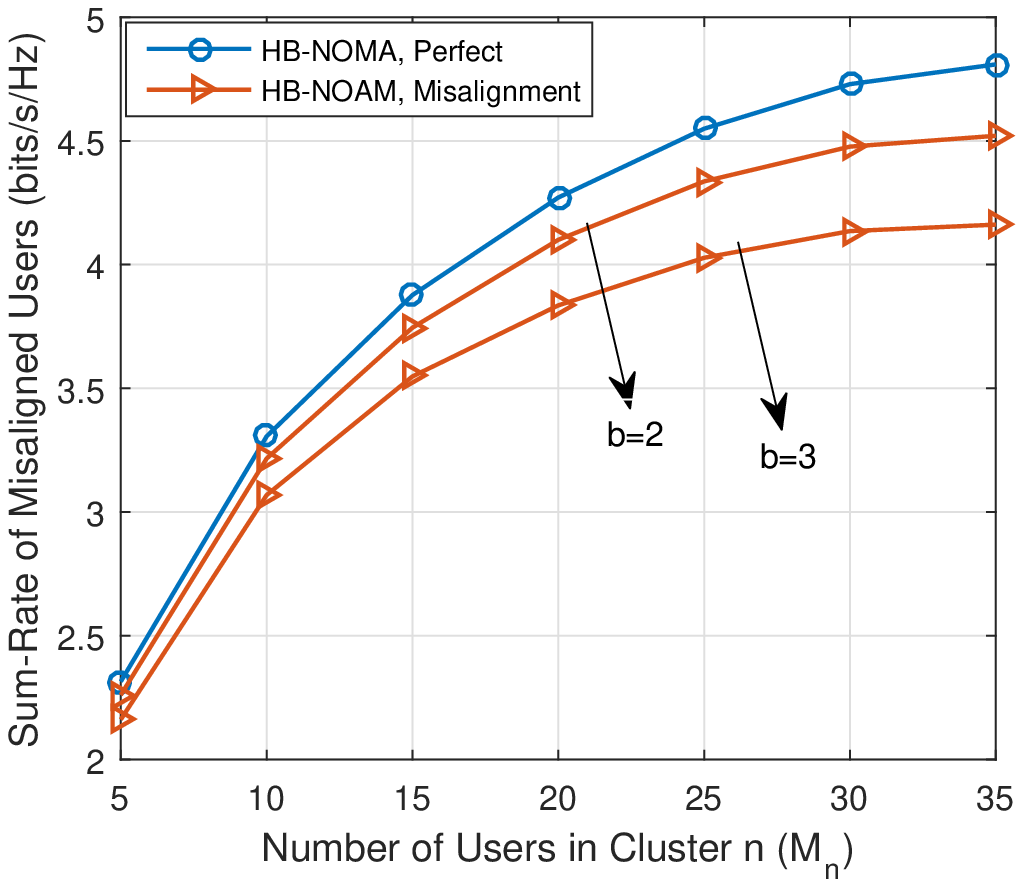}}
        \caption{}
    \end{subfigure}
    \hfill
    \begin{subfigure}[t]{.45\textwidth}
    \centering
    \raisebox{-\height}{\includegraphics[scale=0.6]{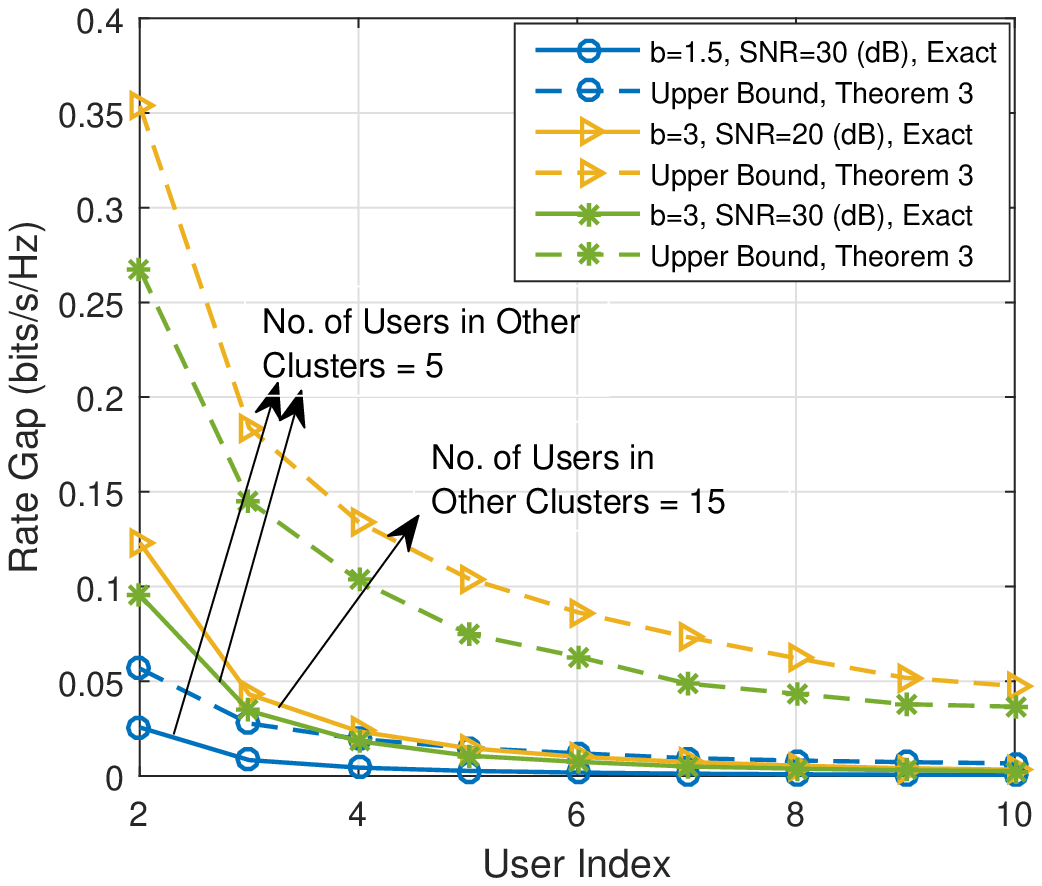}}
    \caption{}
    \end{subfigure}
    \caption{Evaluation of the misalignment on the rate performance of HB-NOMA versus (a) SNR, (b) user index, and (c) number of users per cluster ($M_n$). Also, (d) demonstrates the rate gap among the different misaligned users.}
    \label{fig:imperfect}
\end{figure}

The beam misalignment effect is depicted by Fig.~\ref{fig:imperfect}. We consider five clusters in which $\varphi_{1,1}=10^\circ$, $\varphi_{2,1}=30^\circ$, $\varphi_{3,1}=50^\circ$, $\varphi_{4,1}=65^\circ$, and $\varphi_{5,1}=80^\circ$. All simulations have been done for the middle cluster (third cluster) which is likely imposed the same interference from all the other clusters. Also, the channel gain of the strongest user is 0 dB and the next user's gain drops 1 dB. For instance, the channel gain of U$_{n,m}$ is $-(m-1)$ dB. Fig.~\ref{fig:imperfect}(a), (b), and (d) the number of users in the third cluster is 10.

In Fig.~\ref{fig:imperfect}(a) the achievable rate of two misaligned users U$_{3,2}$ (the strong user) and U$_{3,10}$ (the weak user) versus SNR is shown where the channel gains are -1 and -9 dB, respectively. The misalignment parameter is assumed $b=3$. The number of users in all the other clusters is equal to five. Two different observations are obtained. Increasing the SNR leads to a larger rate gap between perfectly aligned and the misaligned HB-NOMA for the strong user, whereas for the weak users both HB-NOMAs achieve almost the same rate for all SNRs. This demonstrates that the effect of misalignment on the strong users is greater than the weak users. In other words, the weak users should deal with the intra-cluster interference while the strong users should deal with the inter-cluster interference. The other observation is that the lower bound is loose for the strong users but tight for the weak user. The observation indicates that our derived normalized effective channel model in Lemma~\ref{lemma:2} is precise for those users which are intra-cluster interference limited. That is, our finding is able to exactly model the intra-cluster interference. However, the loose lower bound for the strong user indicates that the inter-cluster interference is a little inaccurate which is due to approximating an $N-1$ dimensional subspace with one dimensional space provided in Appendix~\ref{app:lemma2}. \\
To gain more details, we have simulated the achievable rate of all the misaligned users for SNR=15 dB in Fig.~\ref{fig:imperfect}(b). Also, the number of users in the other clusters is set to 15. The mentioned two observations can be seen from this figure, too. However, for strong user, the rate gap between the perfect HB-NOMA and misaligned HB-NOMA is smaller than that of Fig.~\ref{fig:imperfect}(a). Another important observation gained form Fig.~\ref{fig:imperfect}(b) is the impact of the power allocation among the clusters. Based on the proposed power allocation scheme in~(\ref{equPowerAllocation}), to achieve higher rate, more power is assigned to the other clusters than the third cluster which causes U$_{3,2}$ to achieve the rate 0.91 bits/s/Hz. Whereas, for the previous scenario more power is allocated to the third cluster which has more users. Therefore, the rate of U$_{3,2}$ is 0.88 bits/s/Hz. This shows that due to the misalignment the strong clusters leads to higher inter-cluster interference.        

Fig.~\ref{fig:imperfect}(c) compares the sum-rate performance of all the misaligned users with the perfectly aligned HB-NOMA users. Likewise Fig.~\ref{fig:imperfect}(b), we set SNR=15 dB and 15 users for all the clusters except the third. The number of users in the third cluster varies from 5 to 35. Notice that the sum-rate is shown only for the misaligned users, e.g., rate of the first user is neglected. By increasing the number of users, the allocated power to the cluster increases. In consequence, the total rate increases. However, the difference between the aligned and misaligned HB-NOMA becomes worse. Although more users in a cluster means more power is allocated to, the number of users which have inter-cluster interference limited increases as well. As a result, it brings about higher rate lost. Indeed, by making the misalignment parameter worse ($b$=6), the rate lost becomes bigger. It can be concluded that to avoid higher rate lost, HB-NOMA needs to schedule equal number of users per cluster to serve.      

The upper bound evaluation for gap rate between the perfect alignment and misalignment is demonstrated by Fig.~\ref{fig:imperfect}(d). The number of users in other clusters is 5 or 15. For SNR=30 dB and $b$=3, the gap is not substantial and the bound is close to the actual value. When $b$ becomes larger, the gap between the stronger users is bigger than the weaker users. When number of the users of the other cluster increases and simultaneously SNR is reduced, only the stronger users' gap increases. To clarify, for U$_{3,2}$ to U$_{3,5}$, the gap becomes larger, while for the remaining users it is unchanged. The bounds for $b$=6 are not very close to the exact rate gap curves. The main reason is that in the deriving process of the bound in the second line of~(\ref{eq20}) in Appendix~\ref{app:theorem3}, the effect of the inter-cluster interference term is skipped. However, for high misalignment values the interference is considerable. This causes the extracted bound to be less accurate for higher misalignment. 

\begin{figure}[t]   \includegraphics[scale=.6]{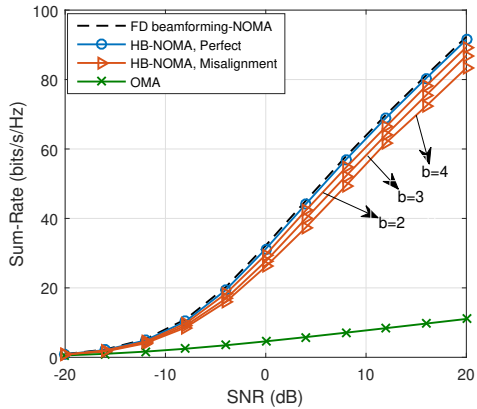}
\centering
        \caption{Sum-rate comparison of the three different systems. The fully-digital and hybrid beamforming systems serve the users using NOMA. The analog system supports the users by exploiting OMA.}
        \label{fig:nomaoma} 
\end{figure}

Our HB-NOMA is compared with the traditional OMA technique in Fig.~\ref{fig:nomaoma}. We choose TDMA for OMA. To gain some insights, three different mmWave systems is evaluated. These systems are fully-digital beamforming, hybrid beamforming and analog beamforming. For fully-digital we assume $N_\text{BS}=N_\text{RF}$=32 which serve 8 clusters. Likewise, for hybrid beamforming we have $N_\text{BS}$=32 but $N_\text{RF}$=8. Both fully-digital and hybrid systems support 8 clusters of users. The first cluster has AoD of $10^\circ$ and AoD of the next clusters increases by $10^\circ$. Further, the users inside of each clusters are distributed in a way that the maximum channel gain difference between the strongest and weakest user is 18 dB. Indeed, the channel gain of the strongest user is 0 dB. The first cluster contains 4 users and each next cluster serves two users more than the previous cluster. Totally, thanks to NOMA technique, both systems support 88 users in each time slot. For OMA, we assume the analog beamforming system equipped with only one RF chain is able to serve one user per time slot. For U$_{n,m}$, the achievable rate of OMA is $\text{log}_2(1+P|\mathbf{w}_{n,m}\mathbf{H}_{n,m}\mathbf{f}_\text{RF}|^2/\sigma^2)$. As expected fully-digital NOMA system achieves the highest sum-rate performance. The HB-NOMA with perfect alignment achieves approximately the same rate as the full-digital. For $b$=2, the misaligned HB-NOMA shows a very close performance to the perfect HB-NOMA. By increasing $b$, the performance slightly decreases. There is a huge rate difference between HB-NOMA and OMA. We conclude that, even in the presence of misalignment, HB-NOMA outperforms OMA.              

 \section{Conclusion}\label{sec:conclusion}
A hybrid beamforming-based NOMA has been designed for the downlink of a single-cell mmWave communication system. To study the achievable rate of an HB-NOMA user, 
we first formulated an optimization problem for the sum-rate of all users in the cell and then proposed an algorithm to solve it in three steps based on the strongest user precoder design. In order to evaluate the sum-rate,   we found  a lower bound for the achievable rate of each user under  perfect and imperfect beam alignment between the effective channel of the users in each cluster. The lower bound analysis demonstrates that perfect HB-NOMA achieves a sum-rate close to that with fully-digital precoder. For the imperfect correlation,   the relationship between the effective channels of the first user and other users inside a cluster was modeled. The bound for the misalignment shows that it is highly function of the mislaigned angle. Such that, a large misalignment angle can cause a significant reduction in the achievable rate. Further,   for each user,   the rate gap between the perfect and imperfect alignment is bounded. The simulation results confirmed our findings. 
\appendices

\section{Proof of Theorem~\ref{theo:1}}\label{app:Theorem1}
\begin{proof}
Given the perfect alignment assumption and (\ref{eq9}),   the effective channel vector for  U$_{n,  m}$ becomes
\begin{align}\label{eq141}
\overbar{\mathbf{h}}_{n,  m}^\dagger &=  \sqrt{N_\text{BS}N_\text{U}}\beta_{n,  m}\mathbf{a}_\text{BS}^\dagger(\varphi_{n,  m})\mathbf{F}_\text{RF} = \beta_{n,  m}\beta_{n,  1}^{-1}\overbar{\mathbf{h}}^\dagger_{n,  1}.
\end{align}
On the other hand,   we have
\begin{equation}\label{eq142}
\overbar{\mathbf{h}}^\dagger_{n,  1}\mathbf{f}^\ell_\text{BB} =
\begin{cases}
\boldsymbol{\Gamma}_{n,  n},   \quad \text{for} \ n,  \ell = 1,   2,   \dots,   N,  \\
0,    \quad \quad \text{ for} \ \ell \neq n.
\end{cases}
\end{equation}
Therefore,   using~(\ref{eq141}) and~(\ref{eq142}) the numerator in~(\ref{eq6}) becomes
\begin{equation}\label{eq143}
    P_{n,  m}\left|\beta_{n,  m}\right|^2\left|\beta_{n,  1}\right|^{-2}\mathbf{\Gamma}_{n,  n}^2.
\end{equation}
Also,   the intra-cluster interference in (\ref{eq61}) becomes $I_\text{intra}^{n,  m} = \sum_{k = 1}^{m-1}P_{n,  k}\left|\beta_{n,  m}\right|^2\left|\beta_{n,  1}\right|^{-2}\mathbf{\Gamma}_{n,  n}^2,$
and the inter-cluster interference term becomes zero,   i.e., $I_\text{inter}^{n,  m} = 0.$

\noindent Now,   substituting~(\ref{eq143}),  and the determined $I_\text{intra}^{n,m}$ and $I_\text{inter}^{n,m}$ in~(\ref{eq6}) gives
\begin{align}\label{eq15}
    \overbar{R}_{n,  m} & = \text{log}_2\left(1 + \frac{P_{n,  m}\left|\beta_{n,  m}\right|^2\left|\beta_{n,  1}\right|^{-2}\mathbf{\Gamma}_{n,  n}^2}{\displaystyle\sum_{k = 1}^{m-1}P_{n,  k}\left|\beta_{n,  m}\right|^2\left|\beta_{n,  1}\right|^{-2}\mathbf{\Gamma}_{n,  n}^2 + \sigma^2}\right) \nonumber \\
    & \overset{(a)}{=} \text{log}_2\left(1 +  \frac{P_{n,  m}N_\text{BS}N_\text{U}\left|\beta_{n,  m}\right|^2}{\displaystyle\sum_{k = 1}^{m-1}P_{n,  k}N_\text{BS}N_\text{U}\left|\beta_{n,  m}\right|^2 + \sigma^2\left(\mathbf{F}^{-1}\right)_{n,  n}}\right)\nonumber \\
    &\overset{(b)}{\geq} \text{log}_2\left(1 +  \frac{P_{n,  m}N_\text{BS}N_\text{U}\left|\beta_{n,  m}\right|^2}{\displaystyle\sum_{k = 1}^{m-1}P_{n,  k}N_\text{BS}N_\text{U}\left|\beta_{n,  m}\right|^2 +\sigma^2 \kappa_\text{min}^{-1}(\mathbf{F})}\right), 
\end{align}
($a$) follows by plugging~(\ref{eq12}) into the expression in the first line of (\ref{eq15}) and using simple manipulations. To get ($b$),   we note  that $\mathbf{F}_\text{RF}$ is full-rank matrix which means $\mathbf{F}=\mathbf{F}_\text{RF}\mathbf{F}_\text{RF}^\dagger$ is positive definite. Then, we have
$\left({{{\mathbf{F}}}}^{-1}\right)_{n,n}\leq\kappa_\text{max}\left({{{\mathbf{F}}}}^{-1}\right)=\kappa_\text{min}^{-1}\left({{{\mathbf{F}}}}\right)$ in which $\kappa_\text{max}(\cdot)$ and $\kappa_\text{min}(\cdot)$ denote the maximum and minimum eigenvalues of $(\cdot)$.  
\end{proof}

\section{Proof of Lemma~1}\label{app:lemma2}
 \begin{proof}
Suppose that the effective channel vectors are fed back by using infinite-resolution codebooks. Also,   let $\hat{\mathbf{h}}_{n,  m}$ denote the normalized effective channel vector for U$_{n,  m}$,   i.e., 
\begin{equation}\label{eq40}
\hat{\tilde{\mathbf{h}}}_{n,  m} = \frac{\tilde{\mathbf{h}}_{n,  m}}{\norm{\tilde{\mathbf{h}}_{n,  m}}}.
\end{equation}

The angle between two complex-valued vectors $\tilde{\mathbf{h}}_{n,m}$ and $ \tilde{\mathbf{h}}_{n,1} \in V_\mathbb{C}$,  denoted by $\Phi_\text{C}$,  is obtained as $
\text{cos}\Phi_\text{C}\overset{\Delta}{=} \rho_{n,m} e^{j\omega_{n,m}} = \hat{\tilde{\mathbf{h}}}_{n,  1}^{\dagger}\hat{\tilde{\mathbf{h}}}_{n,  m},$ where $(\rho_{n,m}\leq 1)$ is equal  to
$\rho_{n,m} = \text{cos}{\Phi}_\text{H}(\hat{\tilde{\mathbf{h}}}_{n,  1},  \hat{\tilde{\mathbf{h}}}_{n,  m}) = \left|\hat{\tilde{\mathbf{h}}}_{n,  1}^{\dagger}\hat{\tilde{\mathbf{h}}}_{n,  m}\right|,$
in which $\Phi_\text{H}(\hat{\tilde{\mathbf{h}}}_{n,  1},  \hat{\tilde{\mathbf{h}}}_{n,  m})$,   $0\leq \Phi_\text{H}\leq\frac{\pi}{2}$,   is the Hermitian angle between two complex-valued vectors $\tilde{\mathbf{h}}_{n,1}$ and $\tilde{\mathbf{h}}_{n,m}$ and $\omega_{n,m}$,   $-\pi\leq\omega_{n,m}\leq\pi$,   is called their pseudo-angle~\cite{scharnhorst2001angles}. The factor $\rho_{n,m}$ describes the angle between the two lines in the complex-valued vector space $V_\mathbb{C}$~\cite{scharnhorst2001angles}.

To ease the analysis, the angle $\omega_{n,m}$ is neglected~\cite{scharnhorst2001angles}. Hence, we find the angle between two lines which are defined by the two vectors $\hat{\tilde{\mathbf{h}}}_{n, 1}$ and $\hat{\tilde{\mathbf{h}}}_{n, m}$. Considering these two vectors as two lines in the space $V_\mathbb{C}$ would be optimistic. However, the simulation results reveal that the derived misalignment model is still effective. Such that, the extracted lower bound for the sum-rate using the misalignment model is close to the exact value of the sum-rate.

For $\ell=n$, the misalignment factor $\rho_{n,m}$ can be calculated as
\begin{align}\label{eqA1}
    \rho_{n,m} \overset{\Delta}{=}  \left|\hat{\tilde{\mathbf{h}}}_{n,  1}^{\dagger}\hat{\tilde{\mathbf{h}}}_{n,  m}\right| &\overset{(a)}{=}
    \frac{N_\text{BS}N_\text{U}\left|\beta_{n,  m}\beta_{n,1}\mathbf{a}_\text{BS}^\dagger(\varphi_{n,  m})\mathbf{F}_\text{RF}\mathbf{F}_\text{RF}^\dagger\mathbf{a}_\text{BS}(\varphi_{n, 1}) \right|}{\norm{\tilde{\mathbf{h}}_{n,  m}}\norm{\tilde{\mathbf{h}}_{n,  1}}} \nonumber \\
    &\overset{(b)}{=}
     \frac{N_\text{BS}N_\text{U}\left|\beta_{n,m}\beta_{n,1}\mathbf{a}_\text{BS}^\dagger(\varphi_{n,  m})\mathbf{V}_1\mathbf{\Lambda}_1\mathbf{V}_1^\dagger\mathbf{a}_\text{BS}(\varphi_{n, 1})\right|}{\norm{\tilde{\mathbf{h}}_{n,  m}}\norm{\tilde{\mathbf{h}}_{n,1}}}\nonumber \\
    &\overset{(c)}{=}
    \frac{\displaystyle\sum_{i=1}^N\kappa_i\left|\mathbf{a}_\text{BS}^\dagger(\varphi_{n,  m})\mathbf{v}_1^i\mathbf{v}_1^{i\dagger}\mathbf{a}_\text{BS}(\varphi_{n, 1})\right|}{\sqrt{\displaystyle\sum_{\ell = 1}^N K_{N_\text{BS}}\left(\varphi_{\ell,  1}-\varphi_{n,m}\right)}\sqrt{\displaystyle\sum_{\ell = 1}^N K_{N_\text{BS}}\left(\varphi_{\ell,  1}-\varphi_{n,1}\right)}}.
\end{align}
To get $(a)$, the expression in~(\ref{eq9}) is used. To get (b), we apply SVD to the Hermitian matrix $\mathbf{F}_\text{RF}\mathbf{F}_\text{RF}^\dagger$ which gives $\mathbf{F}_\text{RF}\mathbf{F}_\text{RF}^\dagger=\mathbf{V}\mathbf{\Lambda}\mathbf{V}^\dagger$ where $\mathbf{V}$ of size $N_\text{BS}\times N_\text{BS}$ is a unitary matrix and $\mathbf{\Lambda}$ of size $N_\text{BS}\times N_\text{BS}$ is a diagonal matrix of singular values ordered in decreasing order. We then partition two matrices $\mathbf{V}$ and $\mathbf{\Lambda}$ as 
\begin{equation}
    \mathbf{V}=\begin{bmatrix}
    \mathbf{V}_1 & \mathbf{V}_2
    \end{bmatrix}
    , \quad \mathbf{\Lambda}=\begin{bmatrix}
    \mathbf{\Lambda}_1 & \mathbf{0} \\
    \mathbf{0} & \mathbf{0}
    \end{bmatrix},
\end{equation}
where $\mathbf{V}_1$ is of size $N_\text{BS}\times N$ and $\mathbf{\Lambda}_1$ and is of size $N\times N$. We note that rank($\mathbf{F}_\text{RF}$)$=N$. Term ($c$) follows from the fact that $\mathbf{\Lambda}_1$ is a diagonal matrix with elements $\kappa_i$ for $i=1, 2, \dots, N$. Notice that $\mathbf{v}_1^i$ represents the $i$th column. 

For $\ell\neq n$, it is reasonable to assume that $\sqrt{1-\rho_{n,m}^2}$ percentage of the amplitude of $\tilde{\mathbf{h}}_{n,m}$ leakages into the subspace generated by the other first users. To determine the subspace, we start with considering  the impact of the misalignment imposed by the other first users on U$_{n,m}$, i.e, $\displaystyle {\sum_{\ell=1,\ell\neq n}^N\left|{\tilde{\mathbf{h}}}_{\ell,  1}^{\dagger}{\tilde{\mathbf{h}}}_{n,  m}\right|^2}$.
Using the definition of vector norm, we rewrite this expression as following: 
\begin{align}
    \sum_{\ell=1,\ell\neq n}^N\left|{\tilde{\mathbf{h}}}_{\ell,  1}^{\dagger}{\tilde{\mathbf{h}}}_{n,  m}\right|^2&=\norm[\Big]{{\tilde{\mathbf{h}}}_{n,  m}^\dagger\begin{bmatrix}{\tilde{\mathbf{h}}}_{1,  1} & \cdots & {\tilde{\mathbf{h}}}_{n-1,  1} & {\tilde{\mathbf{h}}}_{n+1,  1} & \cdots & {\tilde{\mathbf{h}}}_{N,  1}
    \end{bmatrix}}^2\nonumber \\
    &\overset{(a)}{=}{N_\text{BS}N_\text{U}}\norm[\Big]{{\tilde{\mathbf{h}}}_{n,  m}^\dagger\mathbf{F}_\text{RF}^\dagger
    \bigl[\beta_{1,1}\mathbf{a}_\text{BS}\left(\varphi_{1,  1}\right) \text{ } \cdots  \text{ } \beta_{n-1,1}\mathbf{a}_\text{BS}\left(\varphi_{n-1,  1}\right) \nonumber \\  
    & \qquad \qquad \qquad \qquad \qquad  \beta_{n+1,1}\mathbf{a}_\text{BS}\left(\varphi_{n+1,  1}\right) \text{ } \cdots \text{ } \beta_{N,1}\mathbf{a}_\text{BS}\left(\varphi_{N,  1}\right)
    \bigr]}^2\nonumber \\
    &\overset{(b)}{=}{N_\text{BS}N_\text{U}}\norm[\Big]{{\tilde{\mathbf{h}}}_{n,  m}^\dagger\mathbf{F}_\text{RF}^\dagger\mathbf{A}_\text{BS}^{-n}}^2.
\end{align}
To get ($a$), we replace $\tilde{\mathbf{h}}_{\ell,1}$ by~(\ref{eq9}).  Since $\mathbf{a}_\text{BS}\left(\varphi_{n,1}\right)$s are independent vectors, $\mathbf{G}_\text{BS}^{-n}=\sqrt{N_\text{BS}N_\text{U}}\mathbf{F}_\text{RF}^\dagger\mathbf{A}_\text{BS}^{-n}$ determines an $N-1$ dimensional subspace. We represent the weighted linear combination of $\hat{\tilde{\mathbf{h}}}_{\ell,  1}^{\dagger}$ by a new vector $\mathbf{g}_\text{BS}^{-n}$ which is located in the subspace $\mathbf{G}_\text{BS}^{-n}$. So, we get $\mathbf{g}_\text{BS}^{-n}=\sqrt{N_\text{BS}N_\text{U}}\mathbf{F}_\text{RF}\times  \displaystyle\sum_{\ell=1,\ell\neq n}^N\sqrt{P_\ell}\beta_{\ell,1}\mathbf{a}_\text{BS}(\varphi_{\ell,1})$. To get~(\ref{eq19}), we only need to normalize $\mathbf{g}_\text{BS}^{-n}$. 
\end{proof}

\section{Proof of Theorem~2}\label{app:theorem2}
\begin{proof}
Using~(\ref{eq19}),  we obtain the following expressions. First, 
\begin{align}\label{eq191}
\left|\tilde{\mathbf{h}}^{\dagger}_{n,  m}\mathbf{f}_\text{BB}^n\right|^2 &= \rho_{n,  m}^2\norm[\Big]{\tilde{\mathbf{h}}_{n,  m}}^2\left|\hat{\tilde{\mathbf{h}}}^{\dagger}_{n,  1}\mathbf{f}_\text{BB}^n\right|^2  + \left(1 - \rho_{n,  m}^2\right)\norm[\Big]{\tilde{\mathbf{h}}_{n,  m}}^2\left|\mathbf{g}^{-n\dagger}_\text{BS}\mathbf{f}_\text{BB}^n\right|^2 \nonumber \\
& \overset{(a)}{=} \rho_{n,  m}^2\norm[\Big]{\tilde{\mathbf{h}}_{n,  m}}^2\left|\hat{\tilde{\mathbf{h}}}^{\dagger}_{n,  1}\mathbf{f}_\text{BB}^n\right|^2  \overset{(b)}{=} \rho_{n,  m}^2\norm[\Big]{\tilde{\mathbf{h}}_{n,  m}}^2\norm[\Big]{\tilde{\mathbf{h}}_{n,  1}}^{-2}\boldsymbol{\Gamma}_{n,  n}^2, 
\end{align}
in which (a) follows since  $\mathbf{g}^{-n\dagger}_\text{BS}\mathbf{f}_\text{BB}^n = 0$ and (b) follows from~(\ref{eq142}). Second, 
\begin{equation}\label{eq192}
\left|\tilde{\mathbf{h}}^{\dagger}_{n,  m}\mathbf{f}_\text{BB}^\ell\right|^2 =
\left(1-\rho_{n,m}^2\right)\norm[\Big]{\tilde{\mathbf{h}}_{n,  m}}^2\Big|\hat{\mathbf{g}}^{-n\dagger}_\text{BS}\mathbf{f}_\text{BB}^\ell\Big|^2,   \quad \text{for} \quad \ell \neq n.
\end{equation}
Next,   Using~(\ref{eq12}),  ~(\ref{eq142}),  ~(\ref{eq1601}),  ~(\ref{eq40}),   and~(\ref{eq191}),  ~(\ref{eq61}) becomes
\begin{align}\label{eq30}
I_\text{intra}^{n,  m} =& \sum_{k = 1}^{m-1}P_{n,  k}\rho_{n, m}^2N_\text{BS}N_\text{U}\left|\beta_{n,  m}\right|^2\left(\mathbf{F}^{-1}\right)_{n,  n}^{-1}K_{N_\text{BS},  m}K_{N_\text{BS},  1}^{-1}, 
\end{align}
where $K_{N_\text{BS},  1}$ and $K_{N_\text{BS},  m}$ are defined in~(\ref{eq163}).
Likewise,using~(\ref{eq142}),  ~(\ref{eq1601}),  ~(\ref{eq40}),   and~(\ref{eq192}),  ~(\ref{eq62}) becomes
\begin{align}\label{eq31}
I_\text{inter}^{n,  m} = &  \left(1-\rho_{n,m}^2\right)N_\text{BS}N_\text{U}\left|\beta_{n,  m}\right|^2\sum_{\ell \neq n}^NP_\ell\left|\hat{\mathbf{g}}_\text{BS}^{-n\dagger}{\mathbf{f}}_\text{BB}^\ell\right|^2K_{N_\text{BS},  m}.
\end{align}
Further,   after substituting~(\ref{eq191}),  ~(\ref{eq30}) and~(\ref{eq31}) into~(\ref{eq6}),   we get
\begin{align}\label{eq32}
    \tilde R_{n,  m} &= \text{log}_2\left(1+ \frac{\Psi}{I_\text{intra}^{n,  m}+ I_\text{inter}^{n,  m} + \sigma^2}\right)\overset{(a)}{\geq} \text{log}_2\left(1+ \frac{\Psi}{I_\text{intra}^{n,  m}+ \varsigma_\text{inter}^{n,  m} + \sigma^2}\right), 
\end{align}
where $\Psi =P_{n,  m}\rho_{n,m}^2N_\text{BS}N_\text{U}\left|\beta_{n,  m}\right|^2\left(\mathbf{F}^{-1}\right)_{n,  n}^{-1} K_{N_\text{BS},  m}K_{N_\text{BS},  1}^{-1},$
and $\varsigma_\text{inter}^{n,  m} =  \left(1-\rho_{n,m}^2\right)N_\text{BS}N_\text{U}\left|\beta_{n,  m}\right|^2 \times
\kappa_\text{max}(\mathbf{S}) K_{N_\text{BS},  m}$. To get (a),   we have the following lemma.

\begin{lemma}\label{lemma:3}
\normalfont
An upper bound of
$\displaystyle\sum_{\ell=1, \ell \neq n}^NP_\ell\left|\hat{\mathbf{g}}_\text{BS}^{-n\dagger}{\mathbf{f}}_\text{BB}^\ell\right|^2$ is the maximum eigenvalue of $\mathbf{S}$,   i.e.,   $\kappa_\text{max}(\mathbf{S})$.
\end{lemma}
\begin{proof}
We rewrite $\displaystyle\sum_{\ell=1, \ell \neq n}^NP_\ell\left|\hat{\mathbf{g}}_\text{BS}^{-n\dagger}{\mathbf{f}}_\text{BB}^\ell\right|^2 = \norm[\big]{\mathbf{g}_\text{BS}^{-n\dagger}\mathbf{F}_\text{BB}^{-n,W}}^2_2$. Maximizing $\norm[\big]{\hat{\mathbf{g}}_\text{BS}^{-n\dagger}\mathbf{F}_\text{BB}^{-n,W}}^2_2$ given $\norm[\big]{\hat{\mathbf{g}}_\text{BS}^{-n}} = 1$ is similar to maximizing a beamforming vector for maximum ratio transmission systems~\cite{love2003grassmannian,  dighe2003analysis}. Hence,   the maximum value of $\hat{\mathbf{g}}_\text{BS}^{-n}$ is the dominant right singular vector of $\mathbf{F}_\text{BB}^{-n,W}$~\cite{love2003grassmannian,  dighe2003analysis}. Thus,   the maximum of $\norm[\big]{\hat{\mathbf{g}}_\text{BS}^{-n\dagger}\mathbf{F}_\text{BB}^{-n,W}}^2_2$ is equal to the maximum eigenvalue of $\mathbf{S}$.
\end{proof}
Lemma~\ref{lemma:3} indicates that $I_\text{inter}^{n,  m} \leq \varsigma_\text{inter}^{n,  m}$. After some manipulations
\begin{align}\label{eq321}
    \tilde R_{n,  m} & {\geq} \text{log}_2\left(1+ \frac{P_{n,  m}\rho_{n,m}^2N_\text{BS}N_\text{U}\left|\beta_{n,  m}\right|^2}{\zeta_\text{intra}^{n,  m}+ \left(\varsigma_\text{inter}^{n,  m}+ \sigma^2\right) \left(\mathbf{F}^{-1}\right)_{n,  n}K^{-1}_{N_\text{BS},  m}K_{N_\text{BS},  1}}\right)\nonumber \\
    &\overset{(a)}{\geq}\text{log}_2\left(1+ \frac{P_{n,  m}\rho_{n,  m}^2N_\text{BS}N_\text{U}\left|\beta_{n,  m}\right|^2}{\zeta_\text{intra}^{n,  m}+ \zeta_\text{inter}^{n,  m} +  \sigma^2\kappa_\text{min}^{-1}(\mathbf{F})K^{-1}_{N_\text{BS},  m}K_{N_\text{BS},  1}}\right), 
\end{align}
where in the first line,   $\zeta_\text{intra}^{n,  m} = \displaystyle\sum_{k=1}^{m-1}P_{n,  k}\rho_{n,m}^2N_\text{BS}N_\text{U}\left|\beta_{n,  m}\right|^2$ and in the second line, $\zeta_\text{inter}^{n,  m} = \left(1-\rho_{n,m}^2\right)\times N_\text{BS} N_\text{U} \left|\beta_{n,m}\right|^2\kappa_\text{max}(\mathbf{S})\kappa_\text{min}^{-1}(\mathbf{F}) K_{N_\text{BS},1}$. To get (a), we note that $\left(\mathbf{F}^{-1}\right)_{n,n} \leq \kappa_\text{min}^{-1}(\mathbf{F})$.
\end{proof}

\section{Proof of Theorem~\ref{theo:3}}\label{app:theorem3}
\begin{proof}
We start with~(\ref{eq6}) to define the achievable rate of U$_{n,  m}$ for the perfect correlation and the imperfect correlation,   i.e.,   $\overbar{R}_{n,  m}$ and $\tilde{R}_{n,  m}$,   respectively. This gives
\begin{align}\label{eq20}
\Delta R_{n,  m} &\overset{\Delta}{=}  \overbar R_{n,  m} -  \tilde R_{n,  m} \nonumber \\
&=\text{log}_2\left(1 + \frac{P_{n,  m}\left|\overbar{\mathbf{h}}^\dagger_{n,  m}\mathbf{f}_\text{BB}^n\right|^2}{\displaystyle \sum_{k = 1}^{m-1}P_{n,  k}\left|\overbar{\mathbf{h}}^\dagger_{n,  m}\mathbf{f}_\text{BB}^n\right|^2+\sigma^2} \right)   - \nonumber \\
& \qquad \qquad \text{log}_2\left(1 + \frac{P_{n,  m}\left|\tilde{\mathbf{h}}_{n,  m}^{\dagger}{\mathbf{f}}_\text{BB}^n\right|^2}{ \displaystyle \sum_{k=1}^{m-1}P_{n,  k}\left|\tilde{\mathbf{h}}^{\dagger}_{n,  m}{\mathbf{f}}_\text{BB}^n\right|^2 + \sum_{\ell=1, \ell\neq n}^NP_\ell\left|\tilde{\mathbf{h}}^{\dagger}_{n,  m}{\mathbf{f}}_\text{BB}^\ell\right|^2 + \sigma^2}\right)
\nonumber \\
& = \text{log}_2\left( \frac{\displaystyle \sum_{k = 1}^{m} P_{n,  k}\left|\overbar{\mathbf{h}}^\dagger_{n,  m}\mathbf{f}_\text{BB}^n\right|^2 + \sigma^2}{\displaystyle \sum_{k = 1}^{m-1}P_{n,  k}\left|\overbar{\mathbf{h}}^\dagger_{n,  m}\mathbf{f}_\text{BB}^n\right|^2+\sigma^2} \right)  - \text{log}_2\left(\frac{\displaystyle \sum_{k=1}^{m} P_{n,  k}\left|\tilde{\mathbf{h}}_{n,  m}^{\dagger}{\mathbf{f}}_\text{BB}^n\right|^2 + \sum_{\ell=1, \ell\neq n}^NP_\ell\left|\tilde{\mathbf{h}}^{\dagger}_{n,  m}{\mathbf{f}}_\text{BB}^\ell\right|^2 + \sigma^2}{ \displaystyle \sum_{k=1}^{m-1}P_{n,  k}\left|\tilde{\mathbf{h}}^{\dagger}_{n,  m}{\mathbf{f}}_\text{BB}^n\right|^2 + \sum_{\ell=1, \ell\neq n}^NP_\ell\left|\tilde{\mathbf{h}}^{\dagger}_{n,  m}{\mathbf{f}}_\text{BB}^\ell\right|^2 + \sigma^2}\right) \nonumber \\
& \overset{(a)}{\leq} \text{log}_2\left( \frac{\displaystyle \sum_{k = 1}^{m} P_{n,  k}\left|\overbar{\mathbf{h}}^\dagger_{n,  m}\mathbf{f}_\text{BB}^n\right|^2 + \sigma^2}{\displaystyle \sum_{k=1}^{m} P_{n,  k}\left|\tilde{\mathbf{h}}_{n,  m}^{\dagger}{\mathbf{f}}_\text{BB}^n\right|^2 + \sigma^2} \right)   - \text{log}_2\left(\frac{\displaystyle \sum_{k = 1}^{m-1}P_{n,  k}\left|\overbar{\mathbf{h}}^\dagger_{n,  m}\mathbf{f}_\text{BB}^n\right|^2+\sigma^2}{ \displaystyle \sum_{k=1}^{m-1}P_{n,  k}\left|\tilde{\mathbf{h}}^{\dagger}_{n,  m}{\mathbf{f}}_\text{BB}^n\right|^2 + \sum_{\ell=1,\ell\neq n}^NP_\ell\left|\tilde{\mathbf{h}}^{\dagger}_{n,  m}{\mathbf{f}}_\text{BB}^\ell\right|^2 + \sigma^2}\right) \nonumber \\
& \overset{(b)}{\leq} \text{log}_2\left( \frac{\norm[\big]{\overbar{\mathbf{h}}_{n,  m}}^2\left|\hat{\overbar{\mathbf{h}}}^\dagger_{n,  m}\mathbf{f}_\text{BB}^n\right|^2}{\norm[\big]{\tilde{\mathbf{h}}_{n,  m}}^2\left|\hat{\tilde{\mathbf{h}}}_{n,  m}^{\dagger}{\mathbf{f}}_\text{BB}^n\right|^2} \right)  - \text{log}_2\left(\frac{\norm[\big]{\overbar{\mathbf{h}}_{n,  m}}^2\displaystyle \sum_{k = 1}^{m-1}P_{n,  k}\left|\hat{\overbar{\mathbf{h}}}^\dagger_{n,  m}\mathbf{f}_\text{BB}^n\right|^2+1}{\Upsilon}\right),
\end{align}
where $\Upsilon = \norm[\big]{\tilde{\mathbf{h}}_{n,  m}}^2 \displaystyle \sum_{k=1}^{m-1}P_{n,  k}\left|\hat{\tilde{\mathbf{h}}}^{\dagger}_{n,  m}{\mathbf{f}}_\text{BB}^n\right|^2 + \norm[\big]{\tilde{\mathbf{h}}_{n,  m}}^2\sum_{\ell=1, \ell\neq n}^NP_\ell\left|\hat{\tilde{\mathbf{h}}}^{\dagger}_{n,  m}{\mathbf{f}}_\text{BB}^\ell\right|^2 + \sigma^2$.
To get (a) we remove positive quantity $\displaystyle \sum_{\ell=1, \ell\neq n}^NP_\ell\left|\tilde{\mathbf{h}}^{\dagger}_{n,  m}{\mathbf{f}}_\text{BB}^\ell\right|^2$ from the second term. Then,   we exchange the denominator of the first term with the numerator of the second one. (b) follows from the fact that for $u > v$,   it gives $\text{log}\left(\frac{u}{v}\right) > \text{log}\left(\frac{u+c}{v+c}\right)$ ($c>0$),   and  applying the normalized vector $\tilde{\mathbf{h}}_{n,  m}$ defined in~(\ref{eq40}) for both perfect and imperfect effective channel vectors.

Noting that $\hat{\overbar{\mathbf{h}}}_{n,  1} = \hat{\overbar{\mathbf{h}}}_{n,  m}$ and using~(\ref{eq191}) it yields
\begin{align}\label{eq21}
\Delta R
& \leq \text{log}_2\left(\frac{\norm[\big]{\overbar{\mathbf{h}}_{n,  m}}^2}{\rho_{n,  m}^2\norm[\big]{\tilde{\mathbf{h}}_{n,  m}}^2}\right)   - \text{log}_2\left( \displaystyle   \sum_{k=1}^{m-1}P_{n,  k} \norm[\big]{\overbar{\mathbf{h}}_{n,  m}}^2\left|\hat{\tilde{\mathbf{h}}}_{n,  1}^\dagger\mathbf{f}_\text{BB}^n\right|^2+\sigma^2\right) \nonumber \\
& \quad + \text{log}_2\Bigg(\sum_{k = 1}^{m-1}P_{n,  k}\rho_{n,  m}^2\norm[\big]{\tilde{\mathbf{h}}_{n,     m}}^2\left|\hat{\tilde{\mathbf{h}}}_{n,  1}^\dagger{\mathbf{f}}_\text{BB}^n\right|^2 + (1 - \rho_{n,  m}^2)\norm[\big]{\tilde{\mathbf{h}}_{n,  m}}^2\sum_{\ell=1, \ell\neq n}^NP_\ell\left|\hat{\mathbf{g}}^{-n\dagger}_\text{BS}{\mathbf{f}}_\text{BB}^\ell\right|^2 + \sigma^2 \Bigg) \nonumber \\
& \overset{(a)}{=} - \text{log}_2\left( \displaystyle   \sum_{k=1}^{m-1}P_{n,  k}\rho_{n,  m}^2 \left|\hat{\tilde{\mathbf{h}}}_{n,  1}^\dagger\mathbf{f}_\text{BB}^n\right|^2\right)\nonumber \\
&\quad + \text{log}_2\Bigg(\sum_{k = 1}^{m-1}P_{n,  k}\rho_{n,  m}^2\left|\hat{\tilde{\mathbf{h}}}_{n,  1}^\dagger{\mathbf{f}}_\text{BB}^n\right|^2  + (1 - \rho_{n,  m}^2)\sum_{\ell=1, \ell\neq n}^NP_\ell\left|\hat{\mathbf{g}}^{-n\dagger}_\text{BS}{\mathbf{f}}_\text{BB}^\ell\right|^2 + \frac{\sigma^2}{\norm[\big]{\tilde{\mathbf{h}}_{n,  m}}^{2}} \Bigg) \nonumber \\
& \overset{(b)}{\leq} \text{log}_2\left(1 + \frac{\displaystyle \left(1-\rho_{n,  m}^2\right)\kappa_\text{max}\left(\mathbf{S}\right)+\sigma^2\norm[\big]{\tilde{\mathbf{h}}_{n,  m}}^{-2}}{\rho_{n,  m}^2 K^{-1}_{N_\text{BS},1}\left(\mathbf{F}^{-1}\right)^{-1}_{n,n}\displaystyle   \sum_{k=1}^{m-1}P_{n,  k}}\right)\nonumber \\
& \overset{(c)}{\leq} \text{log}_2\left(1 + \frac{\displaystyle \left(1-\rho_{n,  m}^2\right)\kappa_\text{max}\left(\mathbf{S}\right)+\sigma^{2}K^{-1}_{N_\text{BS},m}N_\text{BS}^{-1}N_\text{U}^{-1}\left|\beta_{n,m}\right|^{-2}}{\rho_{n,  m}^2 K^{-1}_{N_\text{BS},1}\kappa_\text{min}(\mathbf{F})\displaystyle   \sum_{k=1}^{m-1}P_{n,  k}}\right),
\end{align}
in which (a) follows by rewriting the first term as $\text{log}_2\left(\rho_{n,  m}^{-2}\norm[\big]{\overbar{\mathbf{h}}_{n,  m}}^2\right) - \text{log}_2\left(\norm[\big]{\tilde{\mathbf{h}}_{n,  m}}^2\right)$. Then,   we sum up the expression $\text{log}_2\left(\rho_{n,  m}^{-2}\norm[\big]{\overbar{\mathbf{h}}_{n,  m}}^2\right)$ with the second term and the expression $-\text{log}_2\left(\norm[\big]{\tilde{\mathbf{h}}_{n,  m}}^2\right)$ with the third term. To get (b), we again sum up the first term with the second term. We then use Lemma~\ref{lemma:3} to get $\kappa_\text{max}(\mathbf{S})$ and~(\ref{eq142}) and~(\ref{eq12}) to get $K^{-1}_{N_\text{BS},1}\left(\mathbf{F}^{-1}\right)^{-1}_{n,n}$. To obtain (c), first we use $\norm[\big]{\tilde{\mathbf{h}}_{n,m}}^2=K_{N_\text{BS},m}N_\text{BS}N_\text{U}|\beta_{n,m}|^2$. Next we use the inequality $\left(\mathbf{F}^{-1}\right)_{n,n} \leq \kappa_\text{min}^{-1}(\mathbf{F})$. 
\end{proof}




\ifCLASSOPTIONcaptionsoff
  \newpage
\fi

\bibliographystyle{IEEEtran}
\bibliography{IEEEabrv,references}
\end{document}